\PassOptionsToPackage{unicode}{hyperref}
\PassOptionsToPackage{hyphens}{url}
\PassOptionsToPackage{dvipsnames,svgnames,x11names}{xcolor}
\documentclass{article}

\usepackage{PRIMEarxiv}
\usepackage{mathptmx}
\usepackage{amsmath,amssymb}
\usepackage{iftex}
\ifPDFTeX
  \usepackage[T1]{fontenc}
  \usepackage[utf8]{inputenc}
  \usepackage{textcomp}
\else
  \usepackage{unicode-math}
  \defaultfontfeatures{Scale=MatchLowercase}
  \defaultfontfeatures[\rmfamily]{Ligatures=TeX,Scale=1}
\fi
\ifPDFTeX\else
\fi
\IfFileExists{upquote.sty}{\usepackage{upquote}}{}
\IfFileExists{microtype.sty}{%
  \usepackage[]{microtype}
  \UseMicrotypeSet[protrusion]{basicmath}
}{}
\makeatletter
\@ifundefined{KOMAClassName}{%
  \IfFileExists{parskip.sty}{%
    \usepackage{parskip}
  }{%
    \setlength{\parindent}{0pt}
    \setlength{\parskip}{6pt plus 2pt minus 1pt}}
}{%
  \KOMAoptions{parskip=half}}
\makeatother
\usepackage{xcolor}
\makeatletter
\ifx\paragraph\undefined\else
  \let\oldparagraph\paragraph
  \renewcommand{\paragraph}{
    \@ifstar
      \xxxParagraphStar
      \xxxParagraphNoStar
  }
  \newcommand{\xxxParagraphStar}[1]{\oldparagraph*{#1}\mbox{}}
  \newcommand{\xxxParagraphNoStar}[1]{\oldparagraph{#1}\mbox{}}
\fi
\ifx\subparagraph\undefined\else
  \let\oldsubparagraph\subparagraph
  \renewcommand{\subparagraph}{
    \@ifstar
      \xxxSubParagraphStar
      \xxxSubParagraphNoStar
  }
  \newcommand{\xxxSubParagraphStar}[1]{\oldsubparagraph*{#1}\mbox{}}
  \newcommand{\xxxSubParagraphNoStar}[1]{\oldsubparagraph{#1}\mbox{}}
\fi
\makeatother

\providecommand{\tightlist}{%
  \setlength{\itemsep}{0pt}\setlength{\parskip}{0pt}}\usepackage{longtable,booktabs,array}
\usepackage{calc}
\usepackage{etoolbox}
\makeatletter
\patchcmd\longtable{\par}{\if@noskipsec\mbox{}\fi\par}{}{}
\makeatother
\IfFileExists{footnotehyper.sty}{\usepackage{footnotehyper}}{\usepackage{footnote}}
\makesavenoteenv{longtable}
\usepackage{graphicx}
\makeatletter
\def\maxwidth{\ifdim\Gin@nat@width>\linewidth\linewidth\else\Gin@nat@width\fi}
\def\maxheight{\ifdim\Gin@nat@height>\textheight\textheight\else\Gin@nat@height\fi}
\makeatother
\setkeys{Gin}{width=\maxwidth,height=\maxheight,keepaspectratio}
\makeatletter
\def\fps@figure{htbp}
\makeatother

\makeatletter
\@ifpackageloaded{caption}{}{\usepackage{caption}}
\AtBeginDocument{%
\ifdefined\contentsname
  \renewcommand*\contentsname{Table of contents}
\else
  \newcommand\contentsname{Table of contents}
\fi
\ifdefined\listfigurename
  \renewcommand*\listfigurename{List of Figures}
\else
  \newcommand\listfigurename{List of Figures}
\fi
\ifdefined\listtablename
  \renewcommand*\listtablename{List of Tables}
\else
  \newcommand\listtablename{List of Tables}
\fi
\ifdefined\figurename
  \renewcommand*\figurename{Figure}
\else
  \newcommand\figurename{Figure}
\fi
\ifdefined\tablename
  \renewcommand*\tablename{Table}
\else
  \newcommand\tablename{Table}
\fi
}
\@ifpackageloaded{float}{}{\usepackage{float}}
\floatstyle{ruled}
\@ifundefined{c@chapter}{\newfloat{codelisting}{h}{lop}}{\newfloat{codelisting}{h}{lop}[chapter]}
\floatname{codelisting}{Listing}

\makeatother
\makeatletter
\@ifpackageloaded{caption}{}{\usepackage{caption}}
\@ifpackageloaded{subcaption}{}{\usepackage{subcaption}}
\makeatother

\ifLuaTeX
  \usepackage{selnolig}
\fi
\IfFileExists{xurl.sty}{\usepackage{xurl}}{}

\usepackage[hidelinks]{hyperref}
\usepackage{bookmark}
\usepackage{titletoc}

\hypersetup{
  pdftitle={Decorated graphons for temporal network estimation},
  pdfauthor={Charles Dufour and Sofia Olhede},
  pdfkeywords={nonparametric estimation, longitudinal relational data, graph processes, stochastic blockmodels, latent structure modeling},
  colorlinks=true,
  linkcolor={blue},
  filecolor={Maroon},
  citecolor={Blue},
  urlcolor={Blue}}

\usepackage[
    backend=biber,
    style=apa,
    natbib=true,
    sortlocale=en_US,
    doi=false,
    eprint=false,
    isbn=true,
    uniquelist=false,
    uniquename=false,
]{biblatex}
\usepackage{enumerate}

\usepackage{url}

\usepackage{mathtools}
\usepackage{amsthm}
\usepackage{amssymb}

\usepackage{bbm}
\usepackage{amsfonts}
\usepackage[capitalise]{cleveref}
\usepackage{enumitem}
\usepackage{tikz}
\usetikzlibrary{arrows, fit, matrix, positioning, shapes, backgrounds, positioning, shapes.geometric, arrows.meta, calc}
\usepackage[most]{tcolorbox}
\usepackage{eqparbox}

\usepackage{thmtools}
\usepackage{thm-restate}

\graphicspath{{figures/}}

\crefname{enumi}{assumption}{assumptions}
\newtheorem{assumption}{Assumption}

\definecolor{oceanblue}{RGB}{33, 113, 181}
\definecolor{burntorange}{RGB}{242, 95, 17}
\definecolor{forestgreen}{RGB}{35, 155, 86}
\definecolor{firebrick}{RGB}{192, 48, 44}
\definecolor{violet}{RGB}{115, 82, 167}
\definecolor{sienna}{RGB}{160, 98, 29}
\definecolor{lavender}{RGB}{194, 165, 207}
\definecolor{slategray}{RGB}{89, 89, 89}
\definecolor{IBM1}{RGB}{100, 143, 255}
\definecolor{IBM2}{RGB}{220, 38, 127}
\definecolor{IBM3}{RGB}{255, 176, 0}

\newcommand{\stext}[1]{\,\text{#1}\,}

\newtheorem{theorem}{Theorem}[section]
\newtheorem{proposition}{Proposition}[section]
\newtheorem{lemma}{Lemma}[section]

\newtheorem{definition}{Definition}[section]
\newtheorem{remark}{Remark}[section]
\newtheorem*{remark*}{Remark}

\theoremstyle{definition}

\newcommand{\indicator}[1]{\mathbbm{1}_{\left\{#1\right\}}}

\newcommand{\probas}[1]{\mathcal{P}\left(#1\right)}
\newcommand{\bbE}{\mathbb{E}}

\newcommand{\N}{\mathbb{N}}

\newcommand{\tx}[2]{#1^{(#2)}}
\newcommand{\e}[2]{\tx{A_{#1}}{#2}}

\newcommand{\est}{\lambda}
\newcommand{\Est}{\Lambda}

\newcommand{\cpsi}{\psi}
\newcommand{\bias}{\beta}
\newcommand{\Eest}{\mu}
\newcommand{\RNum}[1]{%
  \textup{\uppercase\expandafter{\romannumeral#1}}%
}
\newcommand{\UnifDist}{U[0,1]}
\newcommand{\bbP}{\mathbb{P}}

\usepackage{tikz}
\usetikzlibrary{arrows, fit, matrix, positioning, shapes, backgrounds, positioning, shapes.geometric, arrows.meta, calc}
\usepackage[most]{tcolorbox}
\usepackage{xparse}

\definecolor{IBM1}{RGB}{100, 143, 255}
\definecolor{IBM2}{RGB}{220, 38, 127}
\definecolor{IBM3}{RGB}{255, 176, 0}

\begin{document}

\def\spacingset#1{\renewcommand{\baselinestretch}%
{#1}\small\normalsize} \spacingset{1}

\title{Decorated graphons for temporal network estimation}
\author{
	\href{https://orcid.org/0009-0000-3612-7335}{\includegraphics[scale=0.06]{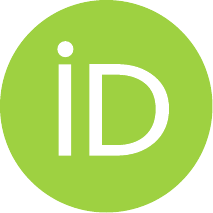}\hspace{1mm}Charles Dufour} \\
	Institute of Mathematics, \'Ecole Polytechnique F\'ed\'erale de Lausanne \\
	Station 8, 1015 Lausanne, Switzerland \\
	\texttt{charles.dufour@epfl.ch}
	\And
	\href{https://orcid.org/0000-0003-0061-227X}{\includegraphics[scale=0.06]{orcid.pdf}\hspace{1mm}Sofia C.~Olhede} \\
	Institute of Mathematics, \'Ecole Polytechnique F\'ed\'erale de Lausanne \\
	Station 8, 1015 Lausanne, Switzerland \\
	\texttt{sofia.olhede@epfl.ch}
}
\maketitle

\begin{abstract}
We propose a unified nonparametric framework for modeling time-evolving networks using decorated graphons (also known as probability-graphons): symmetric functions that assign to each node pair a probability distribution over binary edge time series. This generalizes the static decorated-graphon construction to dynamic graphs while preserving node exchangeability and allowing temporal dynamics such as memory and periodicity. Models in which edges evolve independently given the latent variables, such as autoregressive and Markov edge processes, arise as special cases. We develop a two-stage estimation procedure that separates temporal modeling from network structure. Because the network stage requires only mild regularity conditions on the edge-process estimator, a broad class of temporal edge models can be used in the first stage. We establish nonparametric convergence rates in both block-model and Hölder-smooth regimes, and make explicit how the rate depends on the number of observed time steps and on the quality of the edge-level estimation. We illustrate the method on simulated data and a hospital contact network, recovering latent community structure and time-varying interaction patterns. The framework gives a nonparametric baseline for dynamic network analysis with explicit convergence guarantees.
\end{abstract}


\section{Introduction}\label{sec:intro}

Contact networks observed at high temporal resolution, such as those collected in hospital wards or high schools, often exhibit structured and heterogeneous edge dynamics \citep{vanhems_estimating_2013}. For example, wearable sensors may record face-to-face proximity every minute over several days, resulting in thousands of interactions, typically with strong periodic effects (e.g.\ day-night cycles or school schedules) and considerable heterogeneity in contact durations \citep{stehle_highresolution_2011, fournet_contact_2014}. These data offer valuable insight into processes such as disease transmission and social behavior. Modeling them requires capturing fine-grained, edge-specific temporal patterns while maintaining node exchangeability, which guarantees consistency under subsampling and the existence of a well-defined graph limit \citep{crane_dynamic_2016}. Existing dynamic network models achieve one of these goals at the expense of the other.

Dynamic extensions of classical stochastic blockmodels and latent-space embeddings preserve node-exchangeability by allowing either block-membership or latent coordinates to evolve over time \citep{matias_statistical_2017}. However, these approaches invariably impose strong constraints on temporal dynamics, such as piecewise-constant parameters, conditional independence between time steps, or predetermined change-points, and any available convergence results are proved on a case-by-case basis rather than against a common nonparametric object. Continuous-time  models, by contrast, capture fine-grained interaction patterns and even contagion effects across edges \citep{matias_semiparametric_2018, passino_mutually_2023}, but typically sacrifice exchangeability and require fully parametric specifications with tailored asymptotic theory. More recent proposals, such as block-ALARM \citep{suveges_networks_2023}, autoregressive-edge models \citep{chang_autoregressive_2024, mantziou_gnaredge_2023}, and the AR(1) network model \citep{jiang_autoregressive_2023}, deliver richer edge dynamics or preserve exchangeability, yet none provides a coherent nonparametric limit and an estimation procedure. Consequently, the field lacks a single, well-defined graph-limit object for dynamic networks, a gap we address in this work.

We develop estimation methodology that reconciles node exchangeability with rich temporal structure in dynamic networks. Our approach is based on decorated graphons, in particular $2^\mathbb{N}$‑graphons, which generalize classical graphons by assigning to each node pair a probability law over binary sequences rather than a single edge probability. In the standard setting, a graphon $W$ is a symmetric function from $[0,1]^2$ to $[0,1]$, where each node is associated with a latent coordinate $\xi_i \in [0,1]$, and edges are generated independently with probability $W(\xi_i, \xi_j)$ \citep{lovasz_large_2012}. While decorated graphons have previously been used to model multiplex or attribute-rich networks \citep{dufour_inference_2024}, we show that they can be naturally adapted to capture time-evolving networks with a wide range of edge-wise dynamics. In our formulation, $W(\xi_i, \xi_j)$ is a distribution over infinite binary sequences, and the observed dynamic network is generated by sampling an edge-specific time series from this law for each node pair. This construction retains node exchangeability while allowing for complex temporal behavior, such as periodicity, and other structured dependencies.

Our estimation procedure is modular and proceeds in two stages. In the first stage, each observed edge time series is fit with a flexible model (e.g.\ estimating a Bernoulli rate or a Markov kernel), yielding edge-level parameter summaries. In the second stage, these summaries are aggregated via blockwise least squares to recover an approximation of the decorated graphon. This separation of temporal modeling from network estimation is the key design choice: the theoretical guarantees only require mild regularity conditions on the edge-level estimator, so that any temporal model with an asymptotically consistent parameter estimator can be used in the first stage. We derive convergence rates in both a block-model regime and a Hölder-smooth regime, with error bounds that depend explicitly on the number of time steps and the quality of the edge-level estimation.

We also illustrate the applicability of our approach to a wide range of dynamic network models. For example, the block-ALARM model of \citet{suveges_networks_2023} and the AR$(1)$ model of \citet{jiang_autoregressive_2023} correspond to decorated graphons where $W(\xi_i,\xi_j)$ is a (potentially periodic) Markov chain of higher order. More generally, our framework accommodates a variety of edge processes, as long as mild regularity conditions are met.

By construction, our model does not allow for explicit dependence between edge processes, only through unobserved latent variables. In contrast to \citet{chang_autoregressive_2024,mantziou_gnaredge_2023}, we cannot model explicit triadic or feedback effects across edges. In return, this restriction yields a theoretically tractable framework with interpretable latent structure and a natural nonparametric baseline for comparing more complex models. The remainder of this article is organized as follows. In \Cref{sec:prob_invariances}, we define the decorated graphon model formally. \Cref{sec:estimation} describes the estimation procedure, and \Cref{sec:a_examples} presents theoretical guarantees. In \Cref{sec:application,subsec:hospital}, we provide empirical illustrations on simulated and real datasets. We finish with a discussion of the limitations and potential extensions of our framework in \Cref{sec:conclusion}. The supplemental material contains the proofs and some technical extensions and comments.

\section{From Sequences of Graphs to Decorated Graphs}\label{sec:prob_invariances}

For the sake of exposition, we  consider graph-valued discrete stochastic processes, where each realization manifests as an infinite sequence of aligned simple graphs. However, our observation is restricted to a finite and contiguous subset of this infinite sequence. In practical terms, we observe a sequence of graphs  \(\{G^{(t)}\}\) \(t \in \{1,\ldots,T\}\) on the same node set \(V\), where every \(G^{(t)}\) constitutes a simple graph. \(G^{(t)}\) will be called a \emph{snapshot} at time \(t\). These graphs can be succinctly described using their respective adjacency matrices, denoted as \(A^{(t)} \in \{0,1\}^{n\times n}\), where \(\e{ij}{t}=1\) if the edge \(ij\) was observed at time \(t\), or equivalently \((i,j)\) is an edge in  \(G^{(t)}\), and \(n\) is the number of nodes in the graph.

\begin{figure}[h!]
	\centering
	\begin{tikzpicture}[scale=1, transform shape]

\def\ngon{5}
\def\scalenode{0.8}

\def\starttimelime{0}%
\def\endtimeline{12.5}%
\def\starttimedots{0.2}
\def\endtimedots{12}

\definecolor{colornodes}{RGB}{220, 220, 220}
\def\colorstart{IBM1}
\def\colormid{IBM2}
\def\colorend{IBM3}

\def\firstlist{"\colorstart", "\colormid", "\colorend"}

    \node (snapshots) at (0,0){
	\begin{tikzpicture}[
                show background rectangle,
                scale=1,
                transform shape,
                mystyle/.style={draw,shape=circle,fill=colornodes,scale=\scalenode}
            ]
    	\begin{scope}[on grid, xshift=0.5cm]
                \node[regular polygon,regular polygon sides=\ngon,minimum size=1.5cm] (p) at (1.5,1.5) {};
                \foreach\x in {1,...,\ngon}{\node[mystyle] (n\x) at (p.corner \x){\x};}
                \path[draw, \colorstart , ultra thick] (n1) -- (n2) -- (n3) -- (n1);

            \end{scope}

            \begin{scope}[on grid, xshift=4.5cm]
                \node[regular polygon,regular polygon sides=\ngon,minimum size=1.5cm] (p) at (1.5,1.5) {};
                \foreach\x in {1,...,\ngon}{\node[mystyle] (n\x) at (p.corner \x){\x};}
                \path[draw,\colormid, ultra thick] (n4) -- (n5);
                \path[draw,\colormid, ultra thick] (n2) -- (n3) -- (n1);

            \end{scope}

            \begin{scope}[on grid, xshift=8.5cm]
                \node[regular polygon,regular polygon sides=\ngon,minimum size=1.5cm] (p) at (1.5,1.5) {};
                \foreach\x in {1,...,\ngon}{\node[mystyle] (n\x) at (p.corner \x){\x};}
                \path[draw,ultra thick,\colorend] (n2) -- (n1) -- (n3) -- (n4) -- (n5);
            \end{scope}

            \draw[thick,-latex] (\starttimelime,0) node[xshift=-1.25cm, anchor=west] (t_handle) {$t$} -- (\endtimeline,0);
            \node (G_handle) at (-1cm,1.5) {$G^{(t)}$};
            \node (dots_start) at (\starttimedots,1.5) {\ldots};
            \node (dots_end) at (\endtimedots,1.5) {\ldots};
            \foreach \t [count=\index, evaluate=\index as \c using {{\firstlist}[\index-1]}] in {2,6,10}{
                    \draw[thick, \c] (\t cm,3pt) -- (\t cm,-3pt) node[below] {$\textcolor{\c}{t_{k+\index}}$};
                }
            \node (title) at (6,3) {(a) Infinite sequence of simple graphs};
	\end{tikzpicture}
    };

    \node (decorated-graph) [below=2cm of snapshots, scale=1, transform shape]{
        \begin{tikzpicture}[mystyle/.style={draw,shape=circle,fill=colornodes}, show background rectangle]
            \node[regular polygon,regular polygon sides=\ngon,minimum size=4cm] (p) at (0,2) {};
            \foreach\x in {1,...,\ngon}{\node[mystyle] (n\x) at (p.corner \x){\x};}
            \foreach\x in {1,...,\numexpr\ngon-1\relax}{
              \foreach\y in {\x,...,\ngon}{
                \draw[dashed] (n\x) -- (n\y);
                \node (mid-\x\y) at ($(n\x)!0.5!(n\y)$) {};
              }
            };

            \node[rectangle] (sequence12) at (-3,2) {$[\ldots,\mathbf{\textcolor{\colorstart}{1}},\mathbf{\textcolor{\colormid}{0}},\mathbf{\textcolor{\colorend}{1}}, \ldots]$};
            \draw[{Stealth}-] (sequence12.east) to [in=90,out=0] (mid-12);

             \node[rectangle] (sequence15) at (3,2) {$[\ldots,\mathbf{\textcolor{\colorstart}{0}},\mathbf{\textcolor{\colormid}{0}},\mathbf{\textcolor{\colorend}{0}}, \ldots]$};
             \draw[{Stealth}-] (sequence15.west) to [in=90,out=180] (mid-15);

            \node[rectangle] (sequence23) at (-3,-2) {$[\ldots,\mathbf{\textcolor{\colorstart}{1}},\mathbf{\textcolor{\colormid}{1}},\mathbf{\textcolor{\colorend}{0}}, \ldots]$};
             \draw[{Stealth}-] (sequence23.north) to [in=180,out=90] (mid-23);

             \node[rectangle] (sequence45) at (3,-2) {$[\ldots,\mathbf{\textcolor{\colorstart}{0}},\mathbf{\textcolor{\colormid}{1}},\mathbf{\textcolor{\colorend}{1}}, \ldots]$};
             \draw[{Stealth}-] (sequence45.north) to [in=0,out=90] (mid-45);

            \node (title) at (0,3) {(b) $2^\mathbb{N}$ decorated graph};
	\end{tikzpicture}
    };

    \path[-latex] (snapshots)  edge  node[rectangle, draw, fill=white, dashed,inner sep=0.2cm]  {$ij \longmapsto \left[\ldots,\, A_{ij}^{(\textcolor{\colorstart}{{t_{k+1}}})}, \, A_{ij}^{(\textcolor{\colormid}{t_{k+2}})}, \, A_{ij}^{(\textcolor{\colorend}{t_{k+3}})},\, \ldots\right]$}  (decorated-graph);

\end{tikzpicture}
	\caption{Illustration of changing from an infinite sequence of simple graphs to a \(2^{\N}\)-decorated graph. In (a), we have an infinite sequence of aligned simple graphs. We construct (b) by considering the fully connected graph and mapping each edge \((i,j)\) to the infinite binary sequence encoding its realization in (a). This is the \emph{decoration} of the edge \((i,j)\). The figure shows the mapping for three consecutive times (color-coded) for the edges \((1,2),(2,3),(4,5)\) and \((1,5)\) going counterclockwise from the top left of (b).}\label{fig:decorated_graph}
\end{figure}
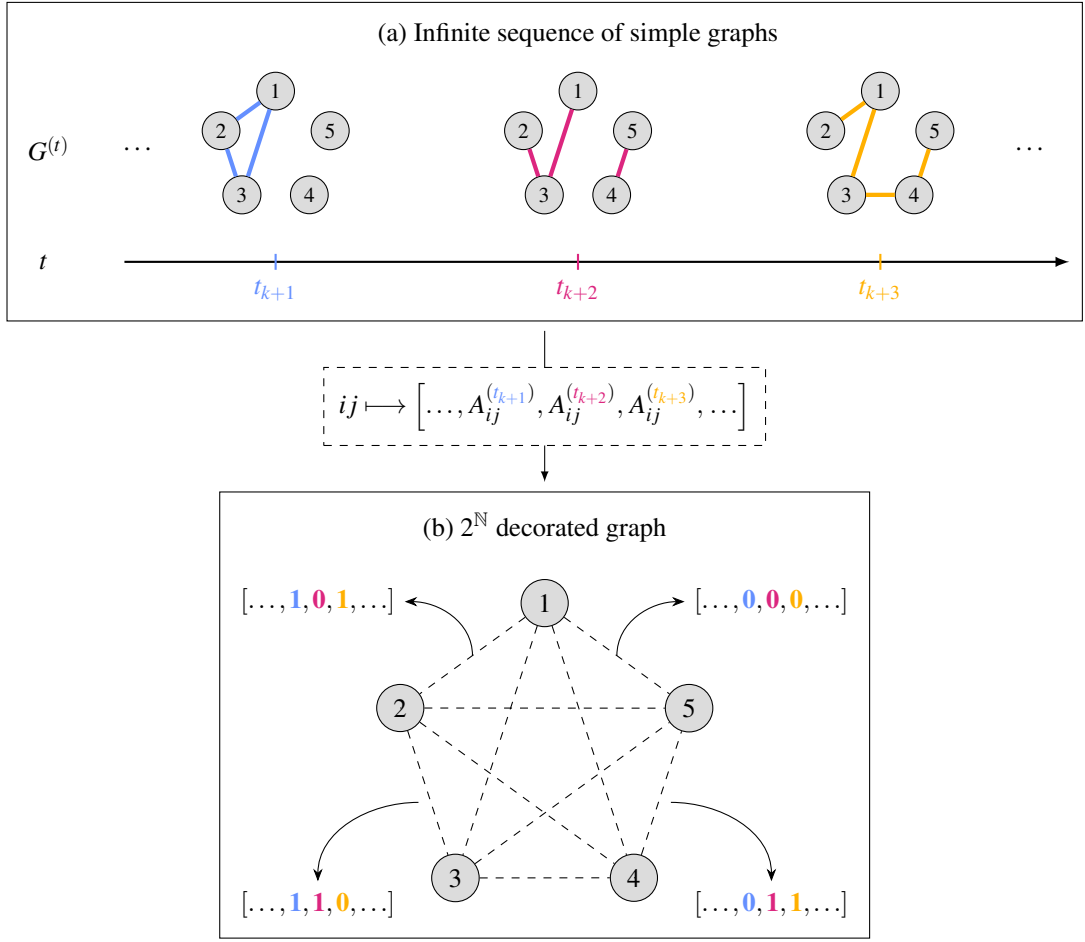

At the heart of our approach is the concept of exchangeability, a foundational assumption in modern network analysis \citep{aldous_representations_1981,hoover_relations_1979,kallenberg_representation_1989}; the joint distribution of the network remains invariant under any permutation of the node labels. This assumption implies that the ordering of the nodes is irrelevant to the underlying stochastic structure, and it has served as a common assumption for many nonparametric network models \citep{olhede_network_2014,orbanz_bayesian_2015,gao_rateoptimal_2015}.

\begin{assumption}[Global exchangeability]
\label{assumption:exchangeability}
The graph-valued process $\{G^{(t)}\}_{t \in \mathbb{N}}$ is jointly exchangeable: for any permutation $\pi$ of the node labels, $\{G^{(t)}\}_{t \in \mathbb{N}} \stackrel{d}{=} \{\pi(G^{(t)})\}_{t \in \mathbb{N}}$, where $\pi(G^{(t)})$ denotes the graph with adjacency matrix $A^{(t)}_{\pi(i)\pi(j)}$.
\end{assumption}

To formally incorporate this invariance, we adopt a decorated graph representation \citep{lovasz_limits_2010}. In this framework, each edge in the complete graph on \(n\) nodes is ``decorated'' with an infinite binary sequence that encodes its temporal evolution as represented in \Cref{fig:decorated_graph} (for more details on decorated graphons, see \citet{dufour_inference_2024} and the references therein). This implies that the law of a discrete graph-valued process can be represented as a \(2^{\N}\)-graphon as defined in \Cref{def:decorated_graphon_ch3}.

\begin{definition}[\(2^{\N}\)-graphon]
	\label{def:decorated_graphon_ch3}
	Let \(\probas{2^{\N}}\) denote the set of Borel probability measures on \(2^{\N}\), where \(2^{\N} = \{0,1\}^{\N}\)  is endowed with the product topology (under which it is compact and metrizable); equivalently, $\probas{2^{\N}}$ is the set of laws of $\{0,1\}$-valued discrete-time stochastic processes. Let \(\mathcal{W}(2^{\N})\) denote the set of Borel-measurable symmetric functions
	\begin{equation}
        \label{eq:decorated_graphon}
		W:[0,1]^2 \to \probas{2^{\N}},
	\end{equation}
    where $\probas{2^{\N}}$ is equipped with the topology of weak convergence. Elements of\ \  \(\mathcal{W}(2^{\N})\) are called \emph{\(2^{\N}\)-graphons}.
\end{definition}

We can think of \(W(\xi_i,\xi_j)\) as providing a distribution for the time-series of edge \((i,j)\). It is then natural to assume that our observed graph sequence was generated by a \(2^{\N}\)-graphon \(W^*\) by the following mechanism. First, for a network with \( n \) nodes, we draw latent variables \( \xi_1, \ldots, \xi_n \) independently from the uniform distribution on \([0,1]\). These latent variables, which remain fixed over time, encode the intrinsic propensities of the nodes to interact and underpin the exchangeability of the network. Given the latent positions, for every unordered pair \((i,j)\) with \(1\leq i < j\leq n\), we sample an infinite binary sequence
\[
A_{ij} = (\e{ij}{1},\e{ij}{2}, \e{ij}{3}, \ldots)
\]
from the discrete binary stochastic process defined by \(W^*(\xi_i, \xi_j) \in \probas{2^{\N}}\). This binary sequence serves as the \emph{decoration} for edge \((i,j)\) and encodes its evolution over time. In practice, for any finite observation window \( t=1,\ldots,T \), the realized states \(\left(\e{ij}{1},\ldots,\e{ij}{T}\right)\) determine whether the edge is active (with \(1\)) or inactive (with \(0\)) at each time point.

Any node-exchangeable discrete graph-valued process can be represented as a \(2^{\N}\)-graphon \citep{kallenberg_probabilistic_2005}. This means that all models that assume conditional edge independence will fall into this framework, including the autoregressive models of \citet{jiang_autoregressive_2023}, the block-ALARM model of \citet{suveges_networks_2023}, and the dynamic network model of \citet{pensky_dynamic_2019} among others (see the supplemental for a more detailed discussion).

\section{A General Framework for Estimating Graph-Valued Processes}\label{sec:estimation}

In this section, we are interested in estimating the \(2^{\N}\)-graphon \(W^*\) from a finite observation window of the graph sequence. We will consider a general framework for estimating the graphon, which is based on the idea of aggregating edge-wise estimates of the probability measures \(W^*(\xi_i,\xi_j)\) for each edge \((i,j)\). The main idea is to decouple the graph aspect from the specificity of the edge decorations, making our estimator applicable to a wide range of models for the edge processes. Additionally, this separation allows for cleaner proofs and provides a plugin framework, where any type of edge process can be accommodated as long as we can estimate their underlying distributions ``well-enough'' (see \cref{assumption:edge_process_est} for a formal statement). This decoupling preserves the convergence rate of static graphon estimation (\Cref{theorem:abstract_result_blockmodel}); any efficiency difference relative to a pooled procedure enters only through the per-edge prefactor $\cpsi^2$ and is discussed in \cref{remark:statistical_efficiency}.

\begin{figure}[h!]
    \centering
    \begin{tikzpicture}[
  node distance=0.5cm and 0.95cm,
  box/.style = {
    draw,
    rounded corners,
    minimum width=2cm,
    minimum height=1.2cm,
    align=center
  },
  arrow/.style = {
    -{Latex[length=2mm]},
    thick
  }
]

\node[box] (input) {$\{A_{ij}^{(t)}\}_{t=1}^{T}$};
\node[box, right=of input] (stage1) {$\lambda_{ij} :=  \Lambda(\{A_{ij}^{(t)}\}) \in \mathcal{X}$};
\node[box, right=of stage1] (stage2) {$\hat{\theta}_{ab} = \frac{1}{\left|\hat{z}^{-1}(a,b)\right|} \sum \lambda_{i j}$};
\node[box, right=of stage2] (output) {$\widehat{W}(x,y) = \hat{\theta}_{\lceil Kx\rceil, \lceil Ky \rceil}$};

\draw[arrow] (input) --  (stage1);
\draw[arrow] (stage1) -- node[above, font=\footnotesize] {L.S}  (stage2);
\draw[arrow] (stage2) -- (output);

\end{tikzpicture}
    \caption{Flowchart of our estimation procedure. The observed edge series $\{A_{ij}^{(t)}\}_{t=1}^T$ are summarized edge-wise into $\est_{ij}$ (\Cref{subsec:assumption:edge_process}), pooled within blocks by the least-squares node clustering (L.S., \Cref{subsec:estimation_graphon}) into block parameters $\hat{\theta}_{ab}$, giving the graphon estimate $\widehat{W}$ as a form of network histogram approximation \citep{olhede_network_2014}. The assumptions place only minimal constraints on $\Est$ and $\theta$, so the methodology applies across different parametrizations of the temporal dynamics.}
    \label{fig:diagram}
\end{figure}
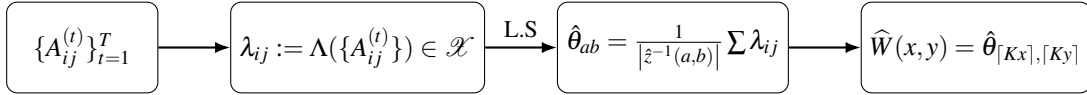

\subsection{Assumption on the Edge Processes}
\label{subsec:assumption:edge_process}
\begin{assumption}
    \label{assumption:parametric_family}
    We assume that we consider an identifiable parametric\footnote{Note that we do not require the parameters to be finite dimensional} family of \(\mathcal{P}(2^{\mathbb{N}})\) that we denote \(\mathcal{P}_{\mathcal{X}}(2^{\mathbb{N}})\), where each element of \(\mathcal{P}_{\mathcal{X}}(2^{\mathbb{N}})\) is parametrized by a \(\theta \in \mathcal{X}\), where \(\mathcal{X}\) is a convex subset of a real Hilbert space, with inner product \(\langle\cdot,\cdot\rangle_{\mathcal{X}}\) and induced norm \(\|\cdot\|_\mathcal{X}\).
\end{assumption}

\Cref{assumption:parametric_family} is a standard parametric assumption that allows us to estimate the probability measures \(W^*(\xi_i,\xi_j)\) for each edge \((i,j)\) independently. We will denote this edge-wise estimator of \(W^*(\xi_i,\xi_j)\) by \(\est_{ij}\) and the corresponding true parameter by \(\theta_{ij}\). We will assume that \(\est_{ij}\) is a \(\mathcal{X}\)-valued  random variable, which is a function of the edge process%
\(\left(\e{ij}{1},\ldots,\e{ij}{T}\right)\), and satisfies the following properties:

\begin{restatable}{assumption}{assumptionEdgeEst}
    \label{assumption:edge_process_est}
    Under \Cref{assumption:parametric_family}, we suppose that there exists a deterministic function
    \begin{equation}
        \Est  :  \bigcup_{T \in \mathbb{N}} \{0,1\}^T \rightarrow \mathcal{X}
    \end{equation}
  such that if $\left(\e{ij}{1},\ldots,\e{ij}{T}\right)$ follows a distribution parametrized by $\theta_{ij}$, then letting the estimator \[\est_{ij} = \Est\left(\e{ij}{1},\ldots,\e{ij}{T}\right)\]
  we assume that
    \begin{itemize}
        \tightlist
        \item  \(\|\est_{ij}-\bbE[\est_{ij}]\|_\mathcal{X}\) is a sub-Gaussian random variable with sub-Gaussian norm \[\|\est_{ij}-\bbE[\est_{ij}]\|_{\psi_2} := \inf\left\{c>0: \bbE\left[\exp\left(\frac{\|\est_{ij}-\bbE[\est_{ij}]\|_{\mathcal{X}}^2}{c^2}\right)\right] \leq 2\right\}.\]
        \item There exist constants $\cpsi, \bias \geq 0$ such that for all $i,j \in [n]$
            $$\|\est_{ij}-\bbE[\est_{ij}]\|_{\psi_2} \leq \cpsi \quad \text{and} \quad \|\bbE[\est_{ij}] - \theta_{ij}\|_\mathcal{X} \leq \bias.$$
    \end{itemize}
\end{restatable}

\Cref{assumption:edge_process_est} requires that for a given temporal observation, we can have a reasonable estimate of the probability measure \(W^*(\xi_i,\xi_j)\) for each edge \((i,j)\) by only considering the information carried by this edge. Here, $\cpsi$ controls the concentration of the estimator around its mean (the sub-Gaussian norm), while $\bias$ bounds the worst-case bias across all edges.
An example of such function \(\Est\) in the case of a Bernoulli process is the empirical mean of the sequence. In this case, we have that \(\Est(\e{ij}{1},\ldots,\e{ij}{T}) = \frac{1}{T}\sum_{t=1}^T \e{ij}{t}\) (details and more examples can be found in \cref{sec:a_examples}).
If the parameter space \(\mathcal{X}\) is bounded, \Cref{assumption:edge_process_est} is trivially satisfied for any unbiased estimator. This assumption may seem somewhat restrictive, but we will show in \Cref{sec:a_examples} that it is satisfied by some common models in the literature. We note that this condition is sufficient but not known to be necessary, and that the results could be extended to more general estimators.

\subsection{Estimation of the Decorated Graphon}
\label{subsec:estimation_graphon}

We will start by estimating the matrix of parameters \(\theta\) with a decorated block with \(k\) blocks by minimizing a similar least-squares criteria as in the binary graphon literature \citep{gao_rateoptimal_2015}:
\begin{equation}
    \label{eq:mse_loss}
    \mathcal{L}(Q,z) = \sum_{i,j\in[n]} \|Q_{z(i)z(j)} - \est_{ij} \|_\mathcal{X}^2,
\end{equation}
where \(Q \in \mathcal{X}_{sym}^{k\times k}\) is a symmetric matrix of parameters, and \(z:[n] \mapsto [k]\) is a function that assigns each node to a group. Our least squares estimator is then given by \(\hat{\theta}_{ij} = \hat{Q}_{\hat{z}(i)\hat{z}(j)}\), where
\begin{equation}
    \label{eq:def_estimator}
    \hat{Q}, \hat{z} = \underset{Q \in \mathcal{X}_{sym}^{k\times k}, z \in \mathcal{Z}_{n,k}}{\operatorname{argmin}} \sum_{i,j\in[n]} \|Q_{z(i)z(j)} - \est_ {ij} \|_\mathcal{X}^2,
\end{equation}
and where \(\mathcal{Z}_{n,k}\) is the set of functions from \([n]\) to \([k]\) representing all possible node groupings. Essentially, this is assigning nodes to $k$ communities and estimating a $k\times k$ matrix of edge-process parameters ($Q$), by minimizing the sum of squared differences between the observed edge parameters $\est_{ij}$ and the model $Q_{z(i),z(j)}$. This is analogous to fitting a stochastic blockmodel on the estimated edge summary statistics. While optimizing \cref{eq:mse_loss} exactly is NP-hard, in practice one could initialize $z$ via a spectral or clustering method on $(\est_{ij})$ and then perform iterative refinement, as is common in the literature \citep{olhede_network_2014}.

We can now state our main results on the convergence rate of the estimator \(\hat{\theta}\) to the true parameter \(\theta\). We first present the results for the case where the edge processes are generated by a block-model, and then for the general case where the underlying graphon \(W^*\) is Hölder-smooth, i.e.,
\[W \in \mathcal{H}_{\alpha}(M) \Longleftrightarrow \|W(x,y)-W(u,v)\|_\mathcal{X} \leq M\left(|x-u| + |y-v|\right)^\alpha.\]

\begin{restatable}{theorem}{theoremAbstractResultBlockmodel}
    \label{theorem:abstract_result_blockmodel}
    Assume that \(A\) is generated by a block-model, i.e.,  there exists \(k \leq n\), such that \(\theta_{ij} = Q_{z(i),z(j)}\) for some \(Q \in \mathcal{X}^{k \times k}\) and \(z \in \mathcal{Z}_{n,k}\). Moreover, let \(\Est\) satisfy \Cref{assumption:edge_process_est}, and \(\hat{\theta}\) as in \cref{eq:def_estimator}. Then, for any constant \(C'>0\), there exists a constant \(C>0\) only depending on \(C'\), such that with probability at least \(1-\exp \left(-C^{\prime} n \log k\right)\)
    \[
    \frac{1}{n^2}\sum_{i,j\in[n]}\|\hat{\theta}_{ij}-\theta_{ij}\|_\mathcal{X}^2 \leq C \left[\cpsi^2 \left(\frac{k^2}{n^2}+ \frac{\log k}{n}\right) + \bias^2 \right].
    \]
\end{restatable}

The rate in \Cref{theorem:abstract_result_blockmodel} has three components: $k^2/n^2$ for estimating $k^2$ block parameters from $n^2$ edge summaries, $\log k / n$ for clustering $n$ nodes into $k$ groups, and $\bias$ accounting for the bias of the edge-wise estimator. When the bias term is $0$, the first two match the static graphon literature \citep{gao_rateoptimal_2015,klopp_oracle_2017}. The factors $\psi^2$ and $\bias$ capture the per-edge estimation error (in the case of simple graphon estimation, $\cpsi^2$ can be related to $(1-\chi)^2$ in \citet{wu_tractably_2025}, though their constant measures non-exchangeability rather than edge-estimation error); $\cpsi^2$  is equal to $1/T$ for Bernoulli processes (\Cref{prop:bernoulli_process}), and scales as $1 / (T\pi_*)$ for Markov chains (\Cref{theorem:mc_graphon}). In these examples, $\psi^2 \to 0$ as $T \to \infty$, reflecting the intuition that longer observation windows yield better edge summaries and thus better graphon estimates. The bias term $\bias$ captures systematic errors in the edge-wise estimator, which could arise from model misspecification or insufficient temporal data. In the Markov chain example, $\bias$ decays exponentially in $T$ (\Cref{lemma:concentration_MC}), so that $\bias^2$ becomes negligible compared to $\cpsi^2$ for any reasonably long observation window. This new term allows us to explicitly capture the effect of the edge-level bias on the overall graphon estimation error.

\begin{restatable}{theorem}{theoremAbstractResultConv}
    \label{theorem:abstract_result_conv}
    Assume \(W \in \mathcal{H}_{\alpha}(M)\). Let \(\theta\) be a symmetric matrix such that \(\theta_{ij} = W(\xi_i,\xi_j)\) for some iid latent variables \(\xi_i \sim  U[0,1]\). Moreover, let \(\Est\) satisfy \Cref{assumption:edge_process_est}, \(k=\lceil n^{1/(\alpha \wedge 1+1)}\rceil\) and \(\hat{\theta}\) as in \cref{eq:def_estimator}. Then, for any constant \(C'>0\), there exists a constant \(C>0\) only depending on \(C'\), such that
    \[\frac{1}{n^2} \sum_{i, j \in[n]}\left\|\hat{\theta}_{i j}-\theta_{i j}\right\|_\mathcal{X}^2 \leq C\left(\left(1+\cpsi^2\right)n^{-2 \alpha /(\alpha+1)}+\cpsi^2\frac{\log n}{n} + \bias^2 \right)\]
    with probability at least \(1-\exp \left(-C^{\prime} n\log n\right)\), where the probability is jointly over \(\{A_{ij}\}\) and \(\{\xi_i\}\).
\end{restatable}

We see here the same behavior as before, but the improvement we could get from \(\Est\) cannot allow us to go faster than the usual nonparametric rate of \(n^{-2\alpha/(\alpha+1)}\). This is due to the approximation of the continuous graphon $W$ by a piecewise constant function, and thus this term appears as a lower bound in the oracle estimator. On the other hand, if \(\Est\) has really high bias, increasing the size of the graph also cannot help. This is where the separation between the edge process (the use of $\Est$) and the graph aspect (unknown design points $\xi_i$) comes in and prevents us from using the temporal repetitions of the edge process to potentially improve the error of the oracle. We will see in the following section that \(\cpsi^2\) will depend on the number of time points \(T\) observed and the mixing properties of each edge process.

For $\alpha < 1$, the rate $n^{-2\alpha/(\alpha+1)}$ in \Cref{theorem:abstract_result_conv} is in fact minimax optimal for the node-centric component of the estimation error, even in the presence of temporal replication. Static graphon estimation is the special case $T=1$ and $\mathcal{X} = \mathbb{R}$, for which the decorated graphon reduces to a classical graphon, so the minimax lower bound of \citet[Theorem~1.2]{gao_rateoptimal_2015} applies and yields
\begin{equation}
    \label{eq:minimax_floor}
    \inf_{\hat{\theta}} \sup_{W \in \mathcal{H}_\alpha(M)}
    \frac{1}{n^2} \sum_{i,j \in [n]}
    \|\hat{\theta}_{ij} - \theta_{ij}\|_\mathcal{X}^2
    \geq C\, n^{-2\alpha/(\alpha+1)}.
\end{equation}
Additional temporal observations can only help, so this floor persists for any $T$; matching the upper bound in \Cref{theorem:abstract_result_conv}, our estimator is rate-optimal in $n$ throughout the H\"older regime $\alpha < 1$.

The same argument as in \citet[Proposition 3.5]{klopp_oracle_2017} lifts these bounds from the parameter matrix $\theta$ to the graphon itself: the Mean Integrated Squared Error (MISE) between the true graphon $W^*$ and its estimate $\hat{W}$ is controlled by the rate of \Cref{theorem:abstract_result_conv} together with an additional approximation (agnostic) term of order $n^{-\alpha \wedge 1}$.

\begin{remark}[Heterogeneous edge processes]
\label{remark:heterogeneous_edges}
    \Cref{assumption:edge_process_est} requires uniform bounds $\cpsi$ and $\bias$ across all edges. This can be relaxed to edge-dependent bounds $\cpsi_{ij}$ and $\bias_{ij}$, replacing the global worst case by blockwise averages in the main theorems and yielding sharper rates when only a few edges have poor concentration (see \cref{subsec:heterogeneous_extension}).
\end{remark}

\begin{remark}[Statistical efficiency of the two-stage procedure]
    \label{remark:statistical_efficiency}
    A distinction should be made between the proposed methodology and the intuitive approach of estimating parameters directly from grouped edges.
    In the latter, nodes are first clustered and parameters are subsequently estimated by pooling all edge processes within each cluster under the assumption of a shared generating mechanism.
    In contrast, the current study estimates parameters for each edge process independently before performing aggregation.

    In the block-model regime (\Cref{theorem:abstract_result_blockmodel}), these two approaches are asymptotically comparable: within a correctly identified block $(a,b)$, both the pooled estimator and the average of $\lfloor n/k \rfloor^2$ independent edge-level estimates converge at the same rate, since the edge-level estimators are independent and identically distributed conditional on the block assignment.
    There is thus no differential cost to the two-stage procedure in this regime: the prefactor $\cpsi^2$, which governs the variance of each edge-level estimate, enters both approaches symmetrically and sets their common scale.

    In the H\"older-smooth regime, however, pooling becomes fundamentally more delicate: even within a correctly identified block, different edges $(i,j)$ and $(u,v)$ may be governed by different parameters (since $W(\xi_i,\xi_j) \neq W(\xi_u,\xi_v)$ in general), so pooled observations are independent but not identically distributed.
    Quantifying the efficiency gain of a pooled estimator in this setting would require a substantially different analysis; see the related work of \citet{leskela_robust_2026}.
    We return to the practical trade-offs of pooling, and to a possible hybrid, in \Cref{sec:conclusion}.
\end{remark}

\section{Example Applications of Our Framework}
\label{sec:a_examples}

In this section, we will show that for specifying a graph-valued discrete stochastic process, one only has to specify a parametric family of edge processes and an estimator. Then using the result from the previous section, we can estimate the decorated graphon based on this estimator using \Cref{theorem:abstract_result_conv}. The strength of our approach is that any edge process that is a member of a parametric family where we can estimate the parameters \(\theta_{ij}\) independently for each edge \((i,j)\) can be directly plugged into our result and will yield an estimator and the corresponding convergence rates.  In the following subsections, we will focus on some examples of a single parametrized edge process, explicitly define the edge-wise estimators defined with \(\Est\), and show that they satisfy \Cref{assumption:edge_process_est}.

\begin{remark}
    $\Est$ is a deterministic function, and the random variable $\est_{ij}$ is the result of applying this function to a subset of the edge process. The random variable $\est_{ij}$ is then a function of the edge process between $i$ and $j$. When defining $\Est$, we will often directly express it via the observed realization of the edge process, as it makes the notation cleaner.
\end{remark}

In order to make the exposition clearer, we will define a shorthand. Say that a \(2^{\mathbb{N}}\)-graphon is such that its image is contained in a specific subset of discrete binary processes (e.g., Bernoulli processes). Additionally, under \Cref{assumption:parametric_family}, say we have that this subset is parametrizable with elements of \(\mathcal{X}\). We will say that \(W\) is \(\mathcal{X}\)-valued graphon or a \textit{subset type} graphon (e.g., Bernoulli process graphon) if \(W:[0,1]^2\rightarrow \mathcal{P}_{\mathcal{X}}\left(2^{\mathbb{N}}\right)\).

\subsection{Binary Edge Variables}

We consider the case where the edge processes are Bernoulli processes, i.e.,
\[\e{ij}{t} \overset{iid}{\sim} \operatorname{Bernoulli}(\theta_{ij}) \quad \text{and} \quad \theta_{ij} = W(\xi_i,\xi_j) \in [0,1].\]
Note that this process includes identical replication (i.e., aligned graphs with the same underlying latents) of the same simple graphon trivially. We define $\Est$ as the empirical mean of the edge process:
\begin{equation}
    \label{eq:est_bernoulli_process}
    \Est\left(\e{ij}{1},\ldots,\e{ij}{T}\right) := \frac{1}{T}\sum_{t=1}^T\e{ij}{t}.
\end{equation}
Since \(\est_{ij}= \Est\left(\e{ij}{1},\ldots,\e{ij}{T}\right) \) trivially satisfies \Cref{assumption:edge_process_est} with \(\cpsi^2 = 1/T\) and \(\bias=0\), we obtain the following convergence rate for the graphon estimation:
\begin{proposition}
    \label{prop:bernoulli_process}
    Let \(W\) be a Bernoulli process graphon satisfying the assumptions of \Cref{theorem:abstract_result_conv}. Let \(\Est\) be as in \cref{eq:est_bernoulli_process}, namely a sample average. Then, for any constant \(C'>0\), there exists a constant \(C>0\) only depending on \(C'\), such that
    \[\frac{1}{n^2} \sum_{i, j \in[n]}\left(\hat{\theta}_{i j}-\theta_{i j}\right)^2 \leq C\left(n^{-2 \alpha /(\alpha+1)}+\frac{\log n}{nT}\right)\]
    with probability at least \(1-\exp \left(-C^{\prime} n\log n\right)\), where the probability is jointly over \(\{A_{ij}\}\) and \(\{\xi_i\}\).
\end{proposition}

When \(T=1\), we retrieve the optimal nonparametric rate of \citet{gao_rateoptimal_2015}. For any \(T\) (even possibly growing with \(n\)), we still have the usual dependence of the rates of the underlying smoothness of $W$:
\[
    \frac{1}{n^2} \sum_{i, j \in[n]}\left(\hat{\theta}_{i j}-\theta_{i j}\right)^2 \leq  \begin{cases}
        n^{-2\alpha/(\alpha+1)} & \text{if } \alpha < 1 \\
        (nT)^{-1}\log(n) & \text{if } \alpha \geq 1 \\
    \end{cases}
\]
This highlights the fact that the edge-wise estimator \(\est\) can only help with the clustering of the nodes, but not with the estimation error done when approximating a non-smooth \(W\) with a piecewise constant function.

\begin{remark}[Temporal replication reduces the clustering penalty]
\label{remark:clustering_penalty}
For $\alpha \geq 1$, temporal replication has a qualitatively different effect from the regime $\alpha < 1$. In the static case ($T=1$), \citet{gao_rateoptimal_2015} show that the minimax rate saturates at $n^{-1}\log n$, where the logarithmic factor reflects the combinatorial cost of assigning $n$ nodes to $k$ clusters. \Cref{prop:bernoulli_process} shows that this logarithmic penalty is modulated by $T$: for $\alpha \geq 1$, the rate becomes $(nT)^{-1}\log n$, so that temporal replication effectively reduces the clustering difficulty.
\end{remark}

\subsection{\texorpdfstring{\(M\)}{M}-Markov Chain Case}
\label{subsec:m_markov_chain_and_periodic}

In this section, we show that if the edge processes are generated by a binary \(M\)th order Markov chain, we can define an estimator \(\Est\) that satisfies the conditions of \Cref{assumption:edge_process_est}. It is easy to see that by considering the state space as \(\{0,1\}^M\), we can now consider a usual Markov chain \citep{billingsley_statistical_1961,wang_approximating_1992}; see \Cref{fig:2_mc_to_1_mc} for a visual representation.

\begin{figure}[h!]
    \centering
    \input{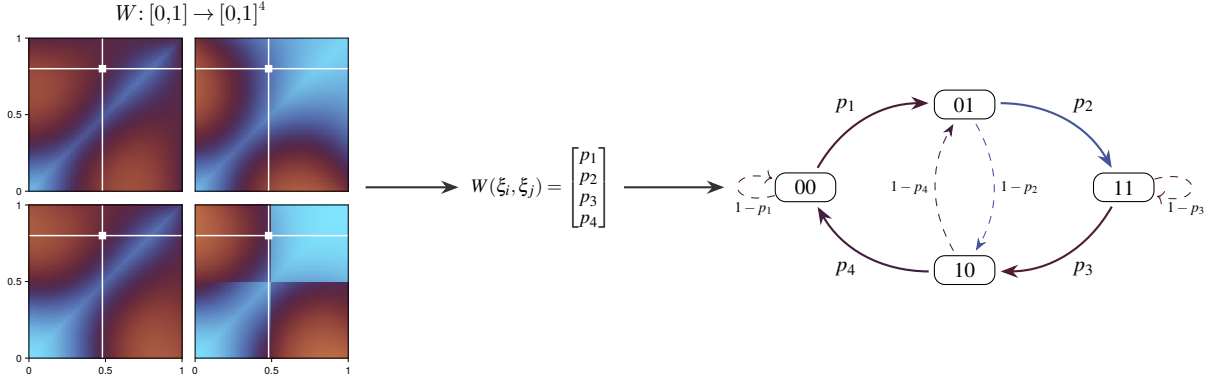}
    \caption{Representation of the decorated graphon where edge dynamics are governed by a $2$-Markov chain $W(\xi_i,\xi_j)$ for each pair of latent variables $(\xi_i, \xi_j)$. We use a De Bruijn graph \citep{debruijn_combinatorial_1946} to represent a second-order Markov chain on \(\{0,1\}\) as a first-order Markov chain on \(\{0,1\}^2\). It is easy to check that any Markov chain represented by a De Bruijn graph is irreducible and aperiodic if all the edges have non-zero weight \citep{kimpton_modelling_2022a}, and that only $4$ parameters are needed to fully specify the transition probabilities of the displayed Markov chain.}
    \label{fig:2_mc_to_1_mc}
\end{figure}

We now show that for a discrete Markov chain with a finite number of states, the empirical transition probabilities are a good estimator of the transition probabilities, in particular that the empirical transition probabilities satisfy \Cref{assumption:edge_process_est} with a concentration bound that depends on the mixing properties of the chain. However, we need some additional regularity condition on the image of the \(M\)-Markov chain graphon for the concentration bounds to hold.

\begin{assumption}
    \label{assumption:mc_graphon}
    Let \(W\) be a \(2^{\N}\)-graphon such that for all $x,y \in [0,1]$,  the edge process governed by \(W(x,y)\) is a binary \(M\)-Markov chain. We assume there exists $\pi_*\in(0,1)$ and \(\tau\in \mathbb{N}\) such that for each \(M\)-Markov chain following the law \(W(x,y)\), we have that
    \begin{itemize}
        \tightlist
        \item The Markov chain is irreducible and aperiodic.
        \item The stationary distribution \(\pi\)  is unique and is lower bounded by \(\pi_*\).
        \item The mixing time of the Markov chain is uniformly bounded by \(\tau\).
    \end{itemize}
    A sufficient condition for the Markov chain to be irreducible and aperiodic is to have all transitions probabilities to be bounded away from $0$. Thus we can parametrize the transition probabilities as stochastic matrices, i.e., \(\mathcal{X} = \{ P \in  [0,1]^{2^M\times 2^M} \mid P \text{ is row stochastic} \}\).
\end{assumption}

\begin{theorem}
    \label{theorem:mc_graphon}
    Let \(W\) be a \(2^{\N}\)-graphon respecting \Cref{assumption:mc_graphon}. Moreover, assume that we observe chains of length at least $T\geq 2^{M+5}\tau^{3/2}\pi_*^{-5/2}/\log\left(2^{M+1}\pi_*^{-1}\right)$. Then under the assumptions of \Cref{theorem:abstract_result_conv}, there exists $\Est$, so that for any constant \(C'>0\), there exists a constant \(C>0\) only depending on \(C'\), \(\pi_*\), \(M\), and \(\tau\) such that
    \begin{restatable}{equation*}{equationMcGraphon}
        \frac{1}{n^2} \sum_{i, j \in[n]}\left\|\hat{\theta}_{i j}-\theta_{i j}\right\|_\mathcal{X}^2 \leq C\left(\left(1+\frac{2^M M}{T\pi_*^{3/2}}\right) n^{-2 \alpha/ (\alpha+1)}+\frac{2^M M}{T\pi_*^{3/2}}\frac{\log n}{n} \right)
    \end{restatable}
    with probability at least \(1-\exp \left(-C^{\prime} n\log n\right)\), where the probability is jointly over \(\{A_{ij}\}\) and \(\{\xi_i\}\) and $\hat{\theta}$ is our least squares estimator defined in \cref{eq:def_estimator}.
\end{theorem}
The condition on $T$ ensures that the empirical transition probabilities concentrate around the true values and that the  bias is absorbed into the sub-Gaussian norm (\Cref{lemma:concentration_MC}).

We now show that the empirical transition probabilities satisfy \Cref{assumption:edge_process_est} with a concentration bound that depends on the mixing properties of the chain. Let us introduce some notation adapted from \citet{wolfer_minimax_2019} to simplify the following discussion. We define $[d]:=\{1, \ldots, d\}$ and use $T$ to denote the length of the observed chain. The simplex of all distributions over $[d]$ will be denoted by $\Delta_d$, and the collection of all $d \times d$ row-stochastic matrices by $\mathcal{P}_d$. For $\mu \in \Delta_d$, we will write either $\mu(i)$ or $\mu_i$, as dictated by convenience (and equivalently for \(P \in \mathcal{P}_d\)). A Markov chain $(P, \mu)$ on $d$ states is specified by an initial distribution $\mu \in \Delta_d$ and a transition matrix $P \in \mathcal{P}_d$ \citep{bremaud_markov_2020}. Namely, $\left(X_1, \ldots, X_T\right) \sim(P, \mu)$ means

\[
\mathbb{P}\left[\left(X_1, \ldots, X_T\right)=\left(x_1, \ldots, x_T\right)\right]=\mu\left(x_1\right) \prod_{t=1}^{T-1} P_{x_t, x_{t+1}}.
\]
Let $\pi$ be the stationary distribution of the Markov chain $(P, \pi)$, and let $\pi_*=\min_{i \in[d]} \pi_i$ be the smallest entry of $\pi$. We will assume that the Markov chain is irreducible and aperiodic (\Cref{assumption:mc_graphon}), which ensures that $\pi$ is unique. Our edge-wise estimator is the matrix of empirical transition probabilities, which coincides with the maximum likelihood estimator \citep{anderson_statistical_1957}:
\begin{equation}
        \label{eq:empirical_transition_probabilities}
        \Est\left(X_1,\ldots,X_T\right) = \hat{P} \in \mathcal{P}_d, \quad \text{where} \quad  \hat{P}_{uv} = \begin{cases}
  \frac{N_{uv}}{N_u} = \frac{\sum_{i=1}^{T-1}\indicator{X_i = u, X_{i+1}=v} }{\sum_{i=1}^{T-1}\indicator{X_i = u} } & \text{if } N_u > 0 \\
\frac{1}{d} & \text{otherwise}.
        \end{cases}
\end{equation}
We now have to show that \(\Est\) satisfies \Cref{assumption:edge_process_est}.

\begin{restatable}{lemma}{lemmaConcentrationMC}
     \label{lemma:concentration_MC}
    Let  \(\left(X_1, \ldots, X_T\right) \sim(P, \mu)\) where \(\mu \in \Delta_d\) and \(P \in \mathcal{P}_d\) with  a unique stationary distribution $\pi$ lower bounded by $\pi_*$ and mixing time $\tau$. Assume \(T > 32d\tau_{\min}^{3/2}\pi_*^{-5/2}/\log\left(2d\pi_*^{-1}\right)\), and let \(\Est\) be the empirical transition probabilities \Cref{eq:empirical_transition_probabilities}.
    Then, \(\Est\) satisfies \Cref{assumption:edge_process_est} with
    \[\cpsi^2 \leq c_1 \frac{d\sqrt{\tau_{\min}}}{T\pi_*^{3/2}}\log\left(\frac{2d}{\pi_*}\right), \quad \text{and} \quad \bias \leq c_2 \frac{d\tau_{\min}}{T\pi_*^2}\]
    for some $c_1,c_2>0$ only depending on the mixing properties of the chain and $\mu$.
\end{restatable}

The proof of \Cref{lemma:concentration_MC} can be found in \cref{subsec:concentration_MC}. We note that our results extend trivially for any graph processes where the edges are decorated with a discrete process taking finitely many possible values at each time point. The main requirement is that the edge processes are generated by a Markov chain with a finite number of states and a unique stationary distribution.

\subsubsection{Periodically Time Varying Markov Chains.} Contact data often exhibit periodicities (e.g., daily routines; see \cref{subsec:hospital}), so the assumption of time-homogeneous edge dynamics may not hold. To account for this, we model edge processes as periodic Markov chains. We encode the time phase into the state space, turning the periodic process into a time-homogeneous chain on an expanded state space \citep{serfozo_basics_2009}.
Let $S$ be a finite state space and let $P\in\mathbb{N}$ be a fixed period (e.g.\ $P=24$ hours for a daily cycle with hourly observations). We define the extended state space as \(\tilde S = S \times \{0,1,\dots,P-1\}\),
where the second component represents time modulo $P$. We consider a discrete-time Markov chain on $\tilde S$ constructed as follows: given current state $(x,t)$, the next state will be $(x', t+1 \bmod P)$ with $x' \sim P_t(x, \cdot)$. Here $P_t(x,y)$ is the transition probability of the original process from state $x$ to state $y$ at phase $t$ of the period. By construction, this augmented chain is time-homogeneous on the finite state space $\tilde S$.
This construction allows us to incorporate regular time-dependent variations, and each periodically time-varying edge process can be treated as a time-homogeneous Markov chain on $\tilde{S}$. \Cref{lemma:concentration_MC} and \Cref{theorem:mc_graphon} then apply directly with $d = 2^M P$.

\section{Simulation Study}
\label{sec:application}

\subsection{Synthetic Data: Continuous Graphon with Markov Chain}

In a first set of experiments, we generated each edge's binary time series from a continuous $2$-Markov chain (MC-$2$) graphon on \([0,1]^2\), simulating a homogeneous second-order Markov chain for each pair of latents.  We then applied our two-stage estimator for varying number of nodes $n$ and number of time steps observed $T$, and visualized the resulting \(2\times2\) transition-probability heatmaps in \Cref{fig:simulation_heatmaps}.  The left grid (rows indexed by \(n\), columns by \(T\)) shows that, as either dimension grows, the estimated block-wise transitions become smoother and converge visibly toward the true graphon displayed on the right.

\begin{figure}[h!]
    \centering
    \includegraphics[width=\textwidth]{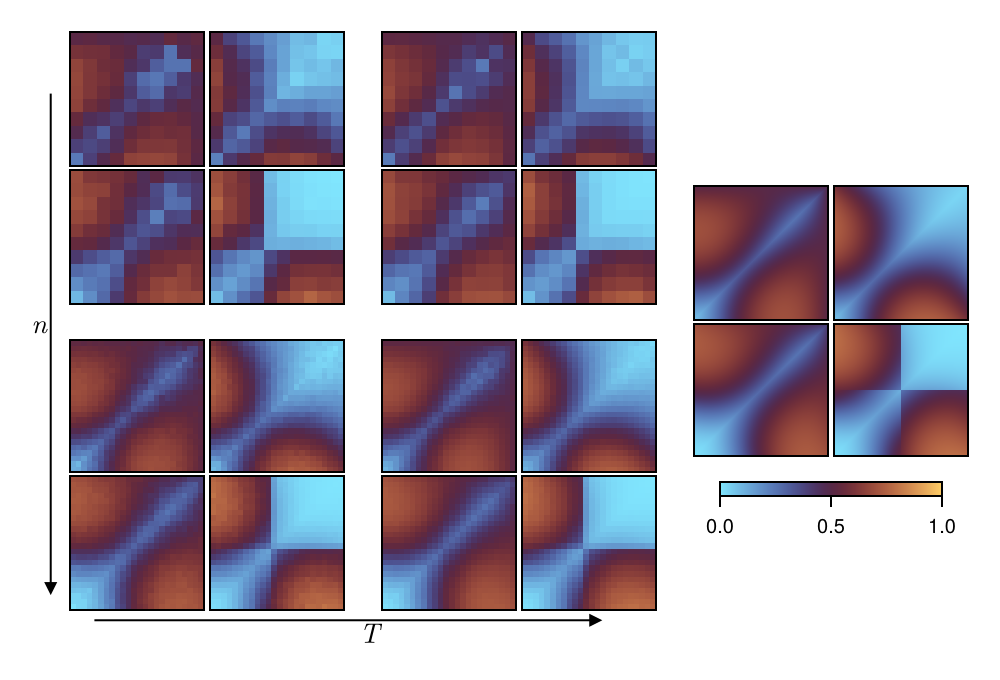}
    \caption{Estimated transition heatmaps (as  depicted in \Cref{fig:2_mc_to_1_mc}) of a continuous MC-$2$ graphon (rows: $n=100,600$; columns: $T=50,300$, and $k=\sqrt{n}$) on the left, and the true MC-$2$ graphon on the right. We can see that as $n$ and $T$ grow, the estimated heatmaps become smoother and closer to the true graphon.}
    \label{fig:simulation_heatmaps}
\end{figure}

To quantify this convergence, we computed the normalized mean-squared error \(n^{-2}\|\hat\theta - \theta\|_2^2\) (similar to \cref{eq:mse_loss}) over ten replicates for \(n\in\{200,400,600\}\) and \(T\in\{50,100,150,200,250,300\}\).
\Cref{fig:simulation_error_vs_n} (left) plots error against \(n\) for each \(T\), and (right) against \(T\) for each \(n\); we also compare different starting node labels for our greedy maximization algorithm: ordered (top) is close to the oracle labeling and random (bottom) is non-informative. In both panels we observe behavior consistent with the theoretical convergence rates of\ \ \Cref{theorem:mc_graphon} (which we repeat here for convenience):
    \equationMcGraphon*

\begin{figure}[h!]
    \centering
    \includegraphics[width=\textwidth]{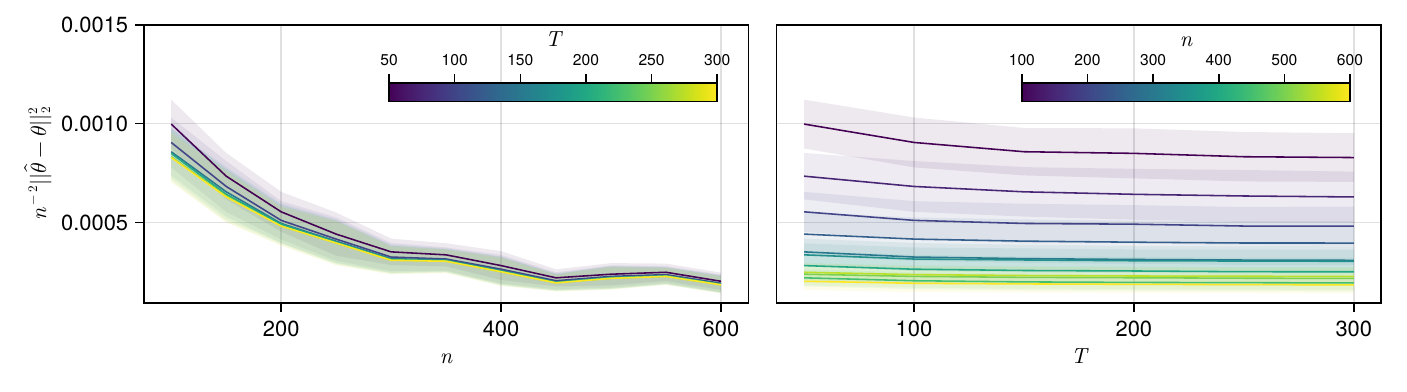}
    \includegraphics[width=\textwidth]{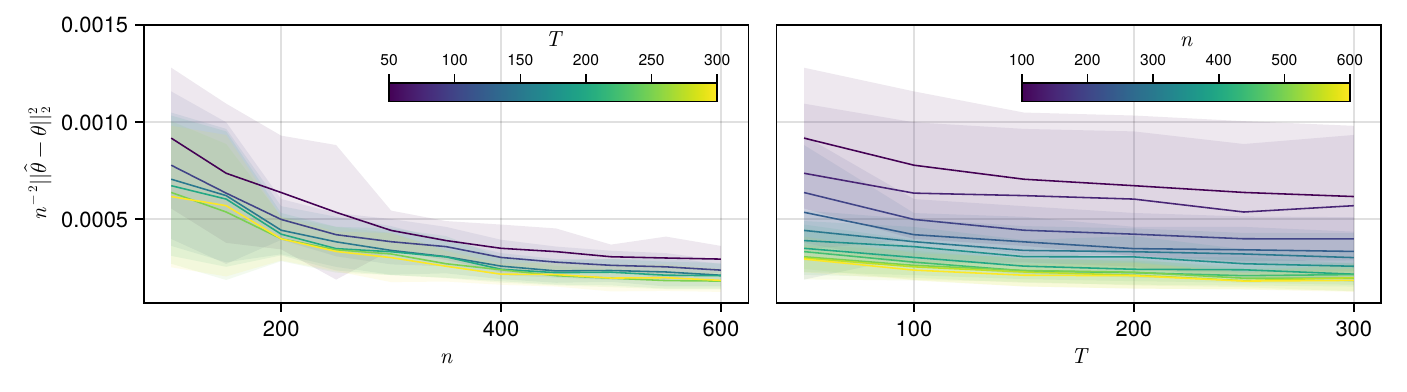}
    \caption{Normalized estimation error $\sum_{i,j}\|\hat{\theta}_{ij}-\theta_{ij}\|_2^2$ plotted against $n$ on the left (curves colored by $T$) and against $T$ on the right (curves colored by $n$). Error bands denote one standard deviation. We used an ordered initialization of the node for the top row and a random initialization for the bottom one; even though the random initialization yields higher variability, the overall trends remain the same, consistent with the theoretical rates.}
    \label{fig:simulation_error_vs_n}
\end{figure}

As \(n\) and \(T\) grow, the estimation error decreases, but for a fixed \(n\), the error does not decrease with \(T\) beyond a certain point. This is consistent with the fact that the estimation error is lower-bounded by \(Cn^{-2\alpha/(\alpha+1)}\) for fixed \(n\). The empirical results of \Cref{fig:simulation_error_vs_n} seem to indicate that this behavior is not only a proof artifact but a real limitation of the estimators.

\subsection{Synthetic Data: Block Model with BALARM Edge Processes}

In practical application, we usually do not have access to the true number of communities $k$ or the smoothness parameter $\alpha$ of the underlying model. To choose the number of blocks fitted (and any other hyperparameters like the order of the underlying Markov chain), we will use the  Bayesian Information Criterion (BIC) \citep{schwarz_estimating_1978} and the Hierarchical BIC \citep{zhao_mixture_2015}. For a discussion of what different kind of criterions represent when picking the number of communities, we refer the reader to \citet{peixoto_descriptive_2023}.

We simulated a dynamic network of \(n=99\) nodes over \(T=960\) time-steps (15-minute intervals for 10 days) from a 3-block BALARM model \citep{suveges_networks_2023}, in which each of the six block-pair edge processes is a second-order autoregressive Markov chain with a daily cosine-sine component (period = 96).  This setup is similar to the one we will encounter with the contact network in \Cref{subsec:hospital}. After randomly permuting node labels, we fitted both first- and second-order edge models for \(k=2,\dots,8\) and used BIC/HBIC to select \(k\) and the Markov order. As shown in Figure~\ref{fig:sim_BALARM}, only the correct second-order specification attains its minimum at \(k=3\), exactly recovering the true block structure, while the first-order fit yields higher scores across all \(k\) but still displays an elbow at \(k=3\).

\begin{figure}[h!]
    \begin{center}
        \includegraphics[width=\textwidth]{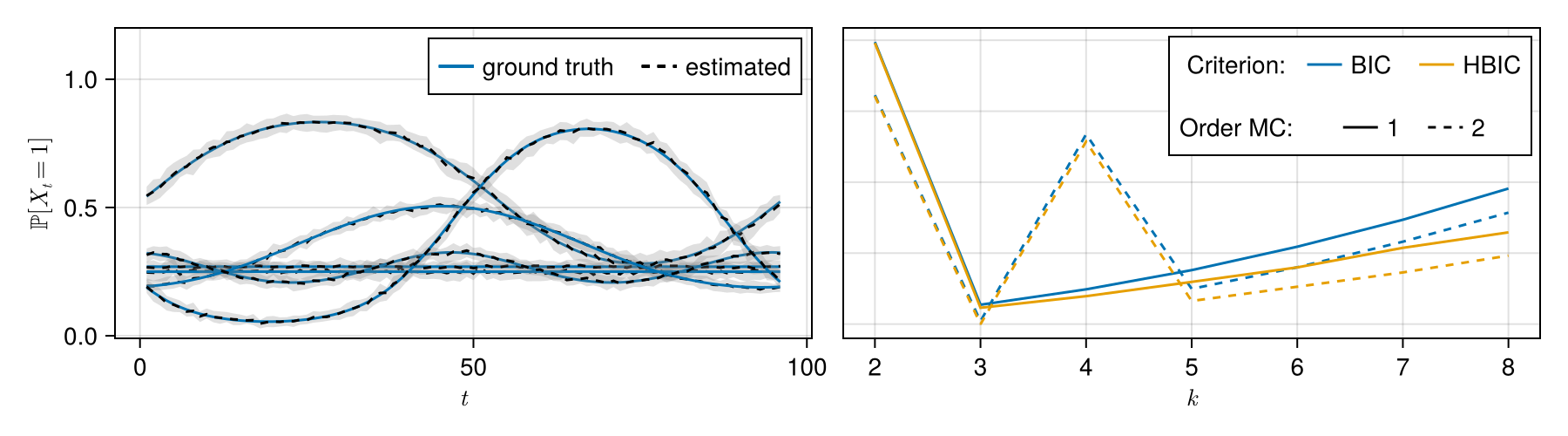}
        \caption{On the left, fitted time varying probability of edge activation for the model picked via the right-hand side plot. We can see that the selected model accurately captures the temporal dynamics of the underlying process. On the right, BIC and HBIC for our simulated data, where the solid line corresponds to the first-order Markov chain and the dashed line to the second-order Markov chain. The minimum of both criteria is at $k=3$, which is the true number of blocks.}
        \label{fig:sim_BALARM}
    \end{center}
\end{figure}

\section{Hospital Ward Dynamic Contact Network}
\label{subsec:hospital}

We apply our method to a high-resolution contact network dataset from a short-stay hospital ward, originally collected by the SocioPatterns collaboration\footnote{\url{http://www.sociopatterns.org/}}. The data captures face-to-face interactions among patients and staff over a four-day period in December 2010 in a geriatric unit of a French hospital in Lyon. Contacts were recorded using wearable proximity sensors with 20-second temporal resolution, yielding a detailed timeline of interactions involving 46 healthcare workers (HCW) and 29 patients (PAT) \citep{vanhems_estimating_2013}. The dataset records every individual in the ward, spanning nurses (NUR), medical doctors (MED), and administrative staff (ADM).
For our analysis, we aggregate the contact data into 15-minute intervals, placing an undirected edge between two individuals if and only if they experienced at least one face-to-face interaction during that period. We restrict attention to pairs who interacted at least once, discarding those with no observed contact.

\Cref{fig:hospital_network} shows how the hospital contact data exhibits clear temporal rhythms (with interactions peaking during daytime hours and repeating daily) and strong heterogeneity in connectivity (a small fraction of individuals accounts for a large share of contacts). The analyses of this dataset by \citet{suveges_networks_2023} have revealed heterogeneous contact duration, role-based mixing patterns (e.g., MED-ADM vs.\ NUR-PAT) resulting in $6$ different contact behaviors (or edge groups). To compare our method with theirs, we fit a $3$ node communities decorated SBM, which results in $6$ different edge groups. The edge process is a Markov chain of order $1$ with a daily periodicity (as described in \cref{subsec:m_markov_chain_and_periodic} with   $P=96$).

\begin{figure}[h!]
    \centering
    \includegraphics[width=\textwidth, trim={0.65cm 0cm 0.4cm 0cm},clip]{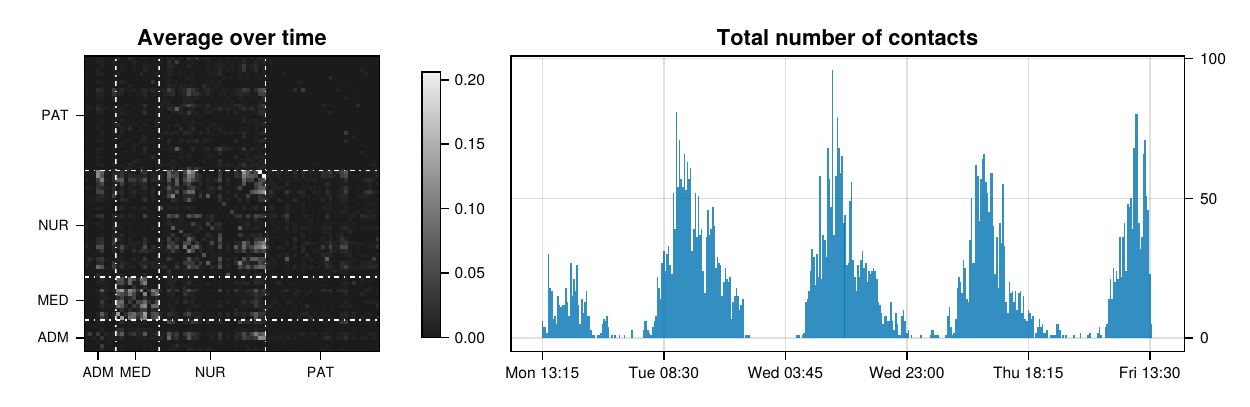}
    \caption{Role-specific contact rates and overall contact volume over time in the hospital ward after the 15-minutes aggregation.
    (Left) Average number of contacts by participant aggregated over the entire dataset: administrative staff (ADM), medical doctors (MED), nurses (NUR), and patients (PAT). (Right) Time series of the total number of contacts across the four-day observation period. The pronounced peaks during daytime hours and troughs at night reflect the ward's daily activity rhythms.}
    \label{fig:hospital_network}
\end{figure}

\begin{figure}[h!]
    \centering
    \includegraphics[width=\textwidth]{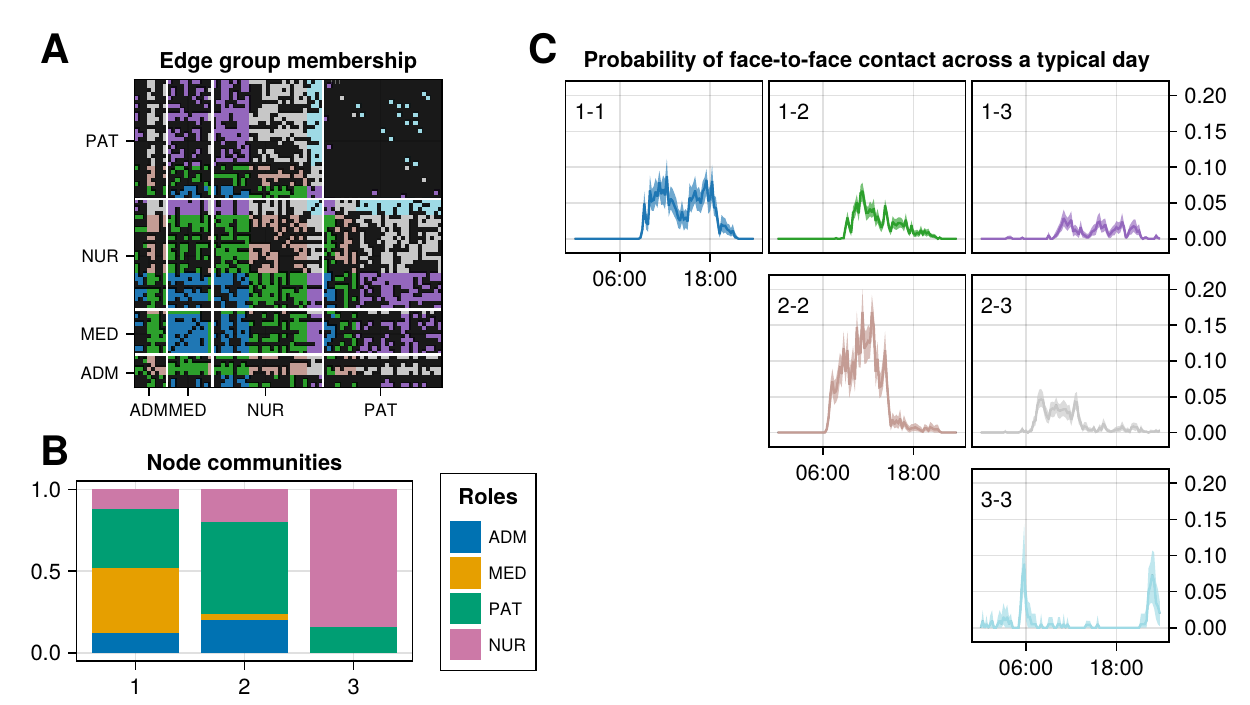}
    \caption{Result of fitting a decorated SBM  with $3$ blocks to the hospital ward contact data with daily periods. Panel A represents each different type of connectivity, where black marks indicate that no contact occurred. Panel B shows the composition of the $3$ node communities in terms of roles. Panel C shows the fitted time varying probability of edge activation and prediction interval obtained via bootstrap with colors matching the edge types in Panel A. The pair of numbers within each subplot indicates which between group connection is being shown.}
    \label{fig:hospital_experiment}
\end{figure}

\Cref{fig:hospital_experiment} presents our inferred block memberships in panels A and B alongside the estimated contact probabilities for each within-block pairing in panel C.  The blue curve (Block 1-1, administrative/medical staff) rises sharply just before 10 AM, peaks around midday, and then declines through the afternoon, closely mirroring the high-probability, morning-peaked (HPM) cluster identified by \citet{suveges_networks_2023}.  The brown curve (Block 2-2, nurse-nurse interactions) shows a sustained high level over mid-day, corresponding to their high-probability, day-long (HPD) cluster.  In contrast, the teal curve (Block 3-3, patient-nurse and patient-patient ties) exhibits a clear bimodal rhythm: an early peak just after 6 AM probably associated with morning rounds and a second peak in the early evening indicating shift handovers.  This resembles the moderate-probability, day-long (MPD) cluster of \citet{suveges_networks_2023} but with a more pronounced evening resurgence. This difference can be attributed to the smooth form of the time-varying probabilities of \citet{suveges_networks_2023}, whereas our method can freely capture any jumps in the time-varying probabilities.  Although our decorated-graphon model uses only three latent node communities, it nonetheless recovers the principal temporal motifs of all six BALARM clusters from \citet{suveges_networks_2023}, while providing finer intraday detail.

\section{Discussion}
\label{sec:extensions_dynamic}
\label{sec:conclusion}

Modeling dynamic networks usually requires a choice between node exchangeability and rich edge dynamics, and convergence guarantees, when they exist, are derived separately for each model. Decorated graphons avoid this choice: any node-exchangeable graph-valued process admits such a representation, so the relevant questions become how well the edge-level law can be estimated and how well a block structure approximates the latent graphon. Our two-stage estimator treats these two questions separately. Because the network stage requires only that the edge-level summaries be accurate enough (\Cref{assumption:edge_process_est}), improvements in modeling individual edge processes feed directly into better graphon estimates without changing the network-level theory. The bounds in \Cref{theorem:abstract_result_conv,theorem:abstract_result_blockmodel} separate a spatial approximation term, fixed by the block structure and bounded below by an irreducible floor (\cref{eq:minimax_floor}), from an edge-estimation term that captures the temporal modeling. Only the second term improves with richer data and better edge models.

The cost of this tractability is the conditional independence of edges given the latent variables (\Cref{assumption:exchangeability}). Triadic closure and feedback across edges therefore lie outside the model, unlike the autoregressive and feedback specifications of \citet{chang_autoregressive_2024,mantziou_gnaredge_2023}. We accept this restriction because ruling out direct inter-edge dependence is what gives the model a well-defined nonparametric limit and convergence guarantees, which the richer dynamic-network models do not currently provide. The estimator is thus best understood as a nonparametric baseline for dynamic networks.

A second limitation is statistical. When $T$ is small or many node pairs rarely interact, the edge-wise summaries $\est_{ij}$ are noisy. Pairs that never interact remain informative, encoding a low activation probability, and can be assigned a degenerate law (similar to a cemetery point as in \citet{abraham_probabilitygraphons_2025}) rather than discarded; our bounds hold uniformly over edges and so cover this case. The more consequential question, raised in \Cref{remark:statistical_efficiency}, is whether to pool. In the block-model regime the two-stage and pooled estimators are asymptotically comparable, so any gain from pooling there is a finite-sample effect, and it is largest precisely when interactions are sparse. Pooling edge parameters within each inferred block pair would reduce variance, at the cost of coupling the edge estimates to the clustering, so that a single misassigned node contaminates an entire block pair. Our two-stage design avoids this coupling but does not exploit within-block homogeneity.

Finally, the modularity of the estimator allows several extensions. Nothing in the network-stage analysis requires $\est_{ij}$ to target the full edge law: \Cref{assumption:edge_process_est} asks only that $\est_{ij}$ be a sub-Gaussian, near-unbiased $\mathcal{X}$-valued summary of the edge process. Any well-behaved summary statistic, be it a finite-dimensional parameter or a vector of moments can therefore be substituted without altering \Cref{theorem:abstract_result_conv,theorem:abstract_result_blockmodel}. The static edge-attributed case is immediate: taking $\est_{ij} = A_{ij}$ and verifying sub-Gaussianity recovers $\mathbb{E}[A_{ij}]$ from \Cref{theorem:abstract_result_conv}, reproducing the graphon estimator of \citet{donier-meroz_graphon_2023}. The main open direction is continuous time. Several models exist there, based on point processes \citep{kreiss_nonparametric_2019,arastuie_chip_2020,modell_intensity_2023,passino_mutually_2023,zhang_semiparametric_2024} or functional parameters \citep{pensky_dynamic_2019}, but a nonparametric limit theory for these processes is still missing, and the difficulty is not only technical: it is not known whether the decorated graphon representation, established for static weighted graphs \citep{abraham_probabilitygraphons_2025,kunszenti-kovacs_multigraph_2022}, extends to continuous-time edge processes. If it does, the estimation strategy developed here would carry over with guarantees of the same form as \Cref{theorem:abstract_result_conv}, again provided edges interact only through their latent variables.

\section*{Acknowledgments}
The authors gratefully acknowledge the European Research Council (ERC) under Grant CoG 2015-682172NETS, within the Seventh European Union Framework Programme.

\spacingset{1}

\setlength{\bibitemsep}{1em}
\printbibliography

\spacingset{1}

\newpage
\setcounter{page}{1}
\begin{center}
    {\LARGE\bf Supplement to \\``Decorated graphons for temporal network estimation''}\\
    \medskip
    {Charles Dufour, Sofia Olhede}
\end{center}
\medskip

\begin{appendix}
\setcounter{equation}{0}
\renewcommand\theequation{A.\arabic{equation}}

\startcontents[sections]
\printcontents[sections]{l}{1}{\setcounter{tocdepth}{2}}

This supplemental material provides theoretical foundations and implementation details for the paper ``Decorated graphons for temporal network estimation''. \Cref{sec:proofs} establishes convergence rates for the least-squares graphon estimator, with proofs for both the block model case (\Cref{subsec:block_model_conv_proof}) and the general Hölder-smooth case (\Cref{sec:holder_convergence}), and \Cref{subsec:bias_estimation} deals explicitly with the bias term. \Cref{sec:appendix-examples} presents proofs for the Markov chain model (\Cref{subsec:concentration_MC}), demonstrating how the general framework applies to specific temporal edge processes. We also provide simulation study details in \Cref{sec:simulation_details}, describing both block model and smooth graphon experimental setups. Finally, \Cref{appendix:technical_extensions} addresses additional technical considerations and comments. The Julia package \texttt{NetworkHistogram.jl} \citep{dufour_networkhistogramjl_2023} implements the methods described in the paper and is available at \url{https://github.com/SDS-EPFL/NetworkHistogram.jl}.

\section{Proofs of probability graphon estimation}
\label{sec:proofs}
We will assume a very general setting, and then show that the specific models presented in the main body of the paper fit in this setting. Let \(\mathcal{K}\) be a compact subset of a Polish space. Let \(\mathcal{P}\left(\mathcal{K}\right)\) denotes the set of probability measures on \(\mathcal{K}\). We will consider an identifiable parametric family of \(\mathcal{P}\left(\mathcal{K}\right)\) that we denote \(\mathcal{P}_{\mathcal{X}}\left(\mathcal{K}\right)\), where each element of \(\mathcal{P}_{\mathcal{X}}\left(\mathcal{K}\right)\) is parametrized by a \(\theta \in \mathcal{X}\), where \(\mathcal{X}\) is a convex subset of a real Hilbert space, with inner product \(\langle\cdot,\cdot\rangle_{\mathcal{X}}\) and induced norm \(\|\cdot\|_\mathcal{X}\).

\begin{remark}
   It would suffice to require that a bijection exists between the set of parameters \(\mathcal{X}\) and the set of probability measures \(\mathcal{P}_{\mathcal{X}}\left(\mathcal{K}\right)\) of interest, allowing for a nonparametric framework, where the parameters $\theta \in \mathcal{X}$ do not have finite dimension. However, since this will not change the results and our proof techniques, we will leave it as a parametrized family for ease of exposition.
\end{remark}

Let \(W^*\) be a \(\mathcal{K}\)-decorated graphon such that its image is in \(\mathcal{P}_{\mathcal{X}}\left(\mathcal{K}\right)\). We suppose we observe \(\{A_{ij}\} \in \mathcal{K}^{n\times n}\) generated from \(W\) with latents \(\{\xi_i\}\). Additionally, we suppose that there exists a function
\begin{equation}
  \Est  : \mathcal{K} \mapsto \mathcal{X} \quad \text{such that} \quad \est_{ij} := \Est(A_{ij})
\end{equation}
is an estimator of \(\theta_{ij} := W(\xi_i,\xi_j)\). This function \(\Est\) is estimating the underlying distribution of \(A_{ij}\) for a fixed pair \(i,j\) by only using that realization and not pulling information from other edges.

We remind the reader of the following assumptions on \(\Est\):
\assumptionEdgeEst*

\begin{remark}
    \label{remark:examples}
    Here are some examples of the settings we consider:
    \begin{itemize}
        \tightlist
        \item  \(\mathcal{K} = \{0,1\}\) and \(A_{ij}\) represents the usual edge indicator in a simple graph. We then take \(\Est(A_{ij}) = A_{ij}\).
        \item \(\mathcal{K} = \{0,1\}^L\), \(A_{ij}\) is a multivariate Bernoulli with independent entries and common mean \(W^*(\xi_i,\xi_j) \in [0,1]\). We can then take \(\Est(A_{ij})\) to be the empirical mean of \(A_{ij}\).
        \item \(\mathcal{K} = \{0,1\}^T\), \(A_{ij}\) is a strongly stationary binary time series generated by a \(M\)-Markov chain with transition probabilities \(W^*(\xi_i,\xi_j) \in [0,1]^{2^M}\). We can then take \(\Est(A_{ij})\) to be the maximum likelihood estimator of the transition probabilities.
    \end{itemize}
\end{remark}

We will start by estimating the matrix of parameters \(\theta\) with a \(k\) block decorated graphon by minimizing the loss function
\begin{equation}
    \mathcal{L}(Q,z) = \sum_{i,j} \|Q_{z(i)z(j)} - \est_{ij} \|_\mathcal{X}^2,
\end{equation}
where \(Q \in \mathcal{X}_{sym}^{k\times k}\) is a symmetric matrix of parameters, and \(z:[n] \mapsto [k]\) is a function that assigns each node to a group. Our least squares estimator is then given by
\begin{equation}
    \label{eq:ols_def}
    \hat{Q}, \hat{z} = \underset{Q \in \mathcal{X}_{sym}^{k\times k}, z \in \mathcal{Z}_{n,k}}{\operatorname{argmin}} \sum_{i,j} \|Q_{z(i)z(j)} - \est_ {ij} \|_\mathcal{X}^2,
\end{equation}
where \(\mathcal{Z}_{n,k}\) is the set of surjective functions from \([n]\) to \([k]\) representing all possible node groupings.

We will be interested in quantifying the distance from our estimator to the true parameter. Given \(\hat{Q},\hat{z}\) defined in \cref{eq:ols_def}, we define
\begin{equation}
    \hat{\theta}_{ij} = \hat{Q}_{\hat{z}(i)\hat{z}(j)},
\end{equation}
and measure its convergence to \(\theta_{ij}\) in terms of the loss function
\begin{equation}
   \|\hat{\theta}-\theta\|^2 := \sum_{i,j} \|\hat{\theta}_{ij} - \theta_{ij} \|_\mathcal{X}^2.
\end{equation}

\begin{remark}
We introduce some shorthand notations: we write \[ \|\hat{\theta}-\theta\|^2 := \sum_{i,j} \|\hat{\theta}_{ij} - \theta_{ij} \|_\mathcal{X}^2, \quad \text{and similarly} \quad  \langle \hat{\theta},\theta \rangle :=  \sum_{i,j} \langle \hat{\theta}_{ij}, \theta_{ij} \rangle_{\mathcal{X}}. \]
\end{remark}

We first establish the theoretical results when the bias of \(\Est\) is $0$ ($\bias=0$) in \Cref{assumption:edge_process_est}, and then extend the results to the case where the bias is non-zero in \Cref{subsec:bias_estimation}. The proof structure for our main results when $\bias=0$ is similar to the one in \citet{gao_rateoptimal_2015}, but our framework removes the explicit dependence on the Bernoulli distribution of the edges. We will first focus on the case where \(W\) is a block model, and then extend the results to the general case.

\subsection{Block Model with Unbiased Edge-wise Estimators}
\label{subsec:block_model_conv_proof}

If we assume that \(W\) is a block model,  \(\theta\) is of the form \(\theta_{ij} = Q^*_{z^*(i)z^*{j}}\) for some \(Q^* \in \mathcal{X}_{sym}^{k\times k}\) and \(z^*\in\mathcal{Z}_{n,k}\). We define \[\tilde{\theta}_{ij} = \tilde{Q}_{\hat{z}(i)\hat{z}(j)}, \text{ with } \tilde{Q}_{ab} = \bar{\theta}_{ab}(\hat{z}) =  \frac{1}{|\hat{z}^{-1}(a)||\hat{z}^{-1}(b)|}\sum_{i \in \hat{z}^{-1}(a)}\sum_{j \in \hat{z}^{-1}(b)}\theta_{ij}.\]
To simplify the proofs, we introduce the following notations for any \(\eta \in \mathcal{X}^{n\times n}\) and \(z \in \mathcal{Z}_{n,k}\):

\[n_a = |z^{-1}(a)| = |\{i: z(i) = a\}|, \quad \text{ and } \quad   \bar{\eta}_{ab}(z) := \frac{1}{n_a n_b}\sum_{i \in z^{-1}(a)}\sum_{j \in z^{-1}(b)}\eta_{ij}.\]

By definition of our estimator, we have that \(\|\hat{\theta}-\est\| \leq \|\theta - \est\|\), and since \(\|\hat{\theta}-\est\|^2 = \|\hat{\theta}-\theta\|^2 + 2 \langle\hat{\theta}-\theta,\theta-\est\rangle + \|\theta-\est\|^2\) by sesquilinearity of the inner product, we have that
\begin{align}
    \frac{1}{2}\|\hat{\theta}-\theta\|^2 &\leq \langle \hat{\theta}-\theta, \est - \theta \rangle \\
    &\leq \|\hat{\theta}-\tilde{\theta}\| \left|\left\langle \frac{\hat{\theta}-\tilde{\theta}}{\|\hat{\theta}-\tilde{\theta}\|}, \est - \theta \right\rangle\right| +  \left(\|\hat{\theta}-\tilde{\theta}\| + \|\hat{\theta}-\theta\| \right)\left|\left\langle \frac{\tilde{\theta}-\theta}{\|\tilde{\theta}-\theta\|}, \est - \theta \right\rangle\right|, \label{eq:UB_loss}
\end{align}
similarly as in \citet{gao_rateoptimal_2015}. We will bound the above using concentration properties of \(\Est\); this will allow us to infer the necessary conditions for the convergence of \(\hat{\theta}\) to \(\theta\), depending on how good \(\est_ {ij}\) is at estimating \(\theta_{ij}\). We will bound each of the following terms with high probability:
\begin{equation}
    \label{eq:terms_UB}
    \underbrace{\|\hat{\theta}-\tilde{\theta}\| \vphantom{\left|\left\langle \frac{\hat{\theta}-\tilde{\theta}}{\|\hat{\theta}-\tilde{\theta}\|}\right\rangle\right|}}_{\RNum{1}:\ \text{\Cref{lemma:concentration_1}}}, \quad
    \underbrace{\left|\left\langle \frac{\tilde{\theta}-\theta}{\|\tilde{\theta}-\theta\|}, \est - \theta \right\rangle\right|}_{\RNum{2}:\ \text{\Cref{lemma:concentration_2}}}, \quad
    \underbrace{\left|\left\langle \frac{\hat{\theta}-\tilde{\theta}}{\|\hat{\theta}-\tilde{\theta}\|}, \est - \theta \right\rangle\right|}_{\RNum{3}:\ \text{\Cref{lemma:concentration_3}}}.
\end{equation}

\theoremAbstractResultBlockmodel*

\Cref{theorem:abstract_result_blockmodel} is similar to the result of \citet{gao_rateoptimal_2015}. The main difference lies in the proper decoupling of
the nature of the estimator \(\Est\) and the structure of the model. This allows our theorem to be applied to a wide range of models, as long as the estimator satisfies \Cref{assumption:edge_process_est}, and not only to the binary case as explored previously in the literature \citep{gao_rateoptimal_2015,klopp_oracle_2017,gaucher_maximum_2021,verdeyme_hybrid_2024}.

\begin{proof}[Proof of\ \ \Cref{theorem:abstract_result_blockmodel} when \(\bias=0\)]
Using \Cref{lemma:concentration_1,lemma:concentration_2,lemma:concentration_3}, we can bound the terms \(\RNum{1},\RNum{2}\), and \(\RNum{3}\) in \cref{eq:terms_UB} with high probability. Combining these bounds with \cref{eq:UB_loss}, we find that
\begin{equation}
    \|\hat{\theta}-\theta\|^2 \leq 2 C\|\hat{\theta}-\theta\| \cpsi\sqrt{n \log k}+4 C^2\cpsi^2\left(k^2+n \log k\right),
\end{equation}
with probability at least \(1-\exp\left(-C' n\log k\right)\). Solving the quadratic equation above, we get
\begin{equation}
    \begin{aligned}
        \|\hat{\theta}-\theta\| & \leq \left(C\cpsi\sqrt{n\log k} + C\cpsi\sqrt{n\log k + 4\left(k^2+n \log k\right) }\right) \\
        & \leq C_1 \cpsi\sqrt{k^2+n \log k},
    \end{aligned}
\end{equation}
for some constant \(C_1>0\), with probability at least \(1-\exp\left(-C' n\log k\right)\), which concludes the proof.
\end{proof}

\begin{lemma}
    \label{lemma:concentration_1}
   For any constant \(C^{\prime}>0\), there exists a constant \(C>0\) only depending on \(C^{\prime}\), such that
    $$
    \|\hat{\theta}-\tilde{\theta}\| \leq C \cpsi\sqrt{k^2+n \log k}
    $$
    with probability at least \(1-\exp \left(-C^{\prime} n \log k\right)\).
\end{lemma}

\begin{proof}[Proof of\ \ \Cref{lemma:concentration_1}]
    \begin{equation}
        \begin{aligned}
        \|\hat{\theta}-\tilde{\theta}\|^2 &= \sum_{i,j}\|\hat{\theta}_{ij}-\tilde{\theta}_{ij}\|_{\mathcal{X}}^2 = \sum_{i,j}\left\|\hat{Q}_{\hat{z}(i)\hat{z}(j)}-\tilde{Q}_{\hat{z}(i)\hat{z}(j)}\right\|_{\mathcal{X}}^2 \\
        &= \sum_{a,b \in [k]}n_a n_b\left\|\bar{\est}_{ab}(\hat{z})-\bar{\theta}_{ab}(\hat{z})\right\|_{\mathcal{X}}^2.
        \end{aligned}
    \end{equation}
    Let \(V_{ab}(z) := n_a n_b\left\|\bar{\est}_{ab}(\hat{z})-\bar{\theta}_{ab}(\hat{z})\right\|_{\mathcal{X}}^2\), we then have
    \begin{equation}
        \label{eq:V_ab}
        \begin{aligned}
            \|\hat{\theta}-\tilde{\theta}\|^2 & \leq \max_{z \in \mathcal{Z}_{n,k}} \sum_{a,b \in [k]} \mathbb{E}\left[V_{ab}(z)\right] +  \max_{z \in \mathcal{Z}_{n,k}} \sum_{a,b \in [k]}\left( V_{ab}(z) - \mathbb{E}\left[V_{ab}(z) \right]\right).
        \end{aligned}
    \end{equation}
    We bound each term separately. For the first term, we have for \(a\neq b\)
    \begin{equation}
        \begin{aligned}
            \mathbb{E}\left[V_{ab}(z)\right] &=  \frac{1}{n_a n_b}\mathbb{E}\left[\left\|\sum_{i \in z^{-1}(a)}\sum_{j \in z^{-1}(b)}\left(\est_ {ij} - \theta_{ij}\right) \right\|_{\mathcal{X}}^2\right] \\
            & = \frac{1}{n_a n_b}\sum_{i \in z^{-1}(a)}\sum_{j \in z^{-1}(b)}\mathbb{E}\left[\left\|\est_ {ij} - \theta_{ij} \right\|_{\mathcal{X}}^2\right] \\
            & \leq \frac{1}{n_a n_b}\sum_{i \in z^{-1}(a)}\sum_{j \in z^{-1}(b)}2\|\est_{ij}-\theta_{ij}\|_{\psi_2}^2 \\
            &\leq 2\cpsi^2,
        \end{aligned}
    \end{equation}
    where we used the conditional independence of \(\est_{ij}\) in the second line and that \(\bias=0\) implies \(\mathbb{E}[\est_{ij}] = \theta_{ij}\) (see \Cref{subsec:bias_estimation} for the adaptation \(\bias\neq 0\)). For \(a=b\), we obtain similar conclusions. We then have that
    \begin{equation}
        \max_{z \in \mathcal{Z}_{n,k}} \sum_{a,b \in [k]} \mathbb{E}\left[V_{ab}(z)\right] \leq C_1 \cpsi k^2,
    \end{equation}
    for some \(C_1>0\). For the second term, using that \(\|\est_{ij}-\theta_{ij}\|_\mathcal{X}\) is a sub-Gaussian random variable and \citet[Prop. 5.10]{vershynin_introduction_2011}, we obtain that \(V_{ab}(z)\) is sub-exponential:
    \begin{equation}
        \begin{aligned}
            \mathbb{P}\left[V_{ab}(z) > t\right] & = \mathbb{P}\left[ \frac{1}{\sqrt{n_a n_b}}\left\|\sum_{i \in z^{-1}(a)}\sum_{j \in z^{-1}(b)}\left(\est_ {ij} - \theta_{ij}\right) \right\|_{\mathcal{X}} > \sqrt{t} \right] \\
            & \leq \mathbb{P}\left[ \frac{1}{\sqrt{n_a n_b}}\sum_{i \in z^{-1}(a)}\sum_{j \in z^{-1}(b)}\left\|\est_ {ij} - \theta_{ij}\right\|_{\mathcal{X}} > \sqrt{t} \right] \\
            & \leq  \mathbb{P}\left[ \frac{1}{n_a n_b}\sum_{i \in z^{-1}(a)}\sum_{j \in z^{-1}(b)}\left\|\est_ {ij} - \theta_{ij}\right\|_{\mathcal{X}} > \sqrt{\frac{t}{n_a n_b}} \right] \\
            & \leq \exp\left(-\frac{ct}{\cpsi^2}\right),
        \end{aligned}
    \end{equation}
    where \(c>0\) is an absolute constant. Using \citet[Prop 5.16]{vershynin_introduction_2011} we obtain
    \begin{equation}
        \mathbb{P}\left[ \sum_{a,b \in [k]}\left( V_{ab}(z) - \mathbb{E}\left[V_{ab}(z) \right]\right)> t\right] \leq \exp\left(-C_2 \min \left\{\frac{t^2}{k^2\cpsi^4},\frac{t}{\cpsi^2}\right\}\right),
    \end{equation}
    for some absolute constant \(C_2>0\). Using a union bound argument, we get
    \begin{equation}
        \mathbb{P}\left[  \max_{z \in \mathcal{Z}_{n,k}} \sum_{a,b \in [k]}\left( V_{ab}(z) - \mathbb{E}\left[V_{ab}(z) \right]\right)> t\right] \leq \exp\left(-C_2 \min \left\{\frac{t^2}{k^2\cpsi^4},\frac{t}{\cpsi^2}\right\}+ n \log k \right).
    \end{equation}
    Thus, for any \(C_3>0\), there exists \(C_4>0\) such that
    \begin{equation}
        \max_{z \in \mathcal{Z}_{n,k}} \sum_{a,b \in [k]}\left( V_{ab}(z) - \mathbb{E}\left[V_{ab}(z) \right]\right) \leq C_3 \left( \cpsi^2n \log k + \sqrt{\cpsi^4nk^2\log k }\right)
    \end{equation}
    with probability at least \(1-\exp \left(-C_4 n\log k\right)\). Plugging these results back in the original inequality \cref{eq:V_ab}, we have
    \begin{equation}
        \begin{aligned}
            \|\hat{\theta}-\tilde{\theta}\|^2 & \leq C_1 \cpsi^2 k^2 + C_3 \left( \cpsi^2n \log k + \sqrt{\cpsi^4nk^2\log k }\right) \\
            & \leq (C_1+2C_3)\cpsi^2 \left(k^2 + n\log k \right)
        \end{aligned}
    \end{equation}
    with probability at least \(1-\exp \left(-C_4 n\log k\right)\), which concludes the proof.
\end{proof}

\begin{lemma}
    \label{lemma:concentration_2}
   For any constant \(C^{\prime}>0\), there exists a constant \(C>0\) only depending on \(C^{\prime}\), such that
    \[
    \left|\left\langle\frac{\tilde{\theta}-\theta}{\|\tilde{\theta}-\theta\|}, \est- \theta\right\rangle\right| \leq C \cpsi\sqrt{ n \log k}
    \]
    with probability at least \(1-\exp \left(-C^{\prime} n \log k\right)\).
\end{lemma}

\begin{proof}[Proof of\ \ \Cref{lemma:concentration_2}]
\label{proof:concentration_2}
    We bound the inner product using a covering argument over the set of possible clusterings $\mathcal{Z}_{n,k}$. Note that $\tilde{\theta}$ is determined entirely by the assignment function $\hat{z}$ and the true parameters $\theta$. For any fixed assignment $z \in \mathcal{Z}_{n,k}$, let $u_z = \frac{\bar{\theta}(z) - \theta}{\|\bar{\theta}(z) - \theta\|}$ be the normalized direction of the bias. The term of interest is bounded by
    \[
    \sup_{z \in \mathcal{Z}_{n,k}} \left|\langle u_z, \est - \theta \rangle\right|.
    \]
    For a fixed $z$, we have
    $$ \langle(u_z)_{ij}, (\est_{ij} - \theta_{ij})\rangle \leq \|(u_z)_{ij}\|\|\est_{ij} - \theta_{ij}\|,$$
    and since by assumption $\|\est_{ij}-\theta_{ij}\|$ is sub-Gaussian, $\langle u_z, \est - \theta \rangle$ is a sum of independent sub-Gaussian random variables with variance proxy at most $\cpsi^2$ (since $\|(u_z)_{ij}\| \leq 1$). Using the standard sub-Gaussian tail bound:
    \[
    \mathbb{P}\left[ \left|\langle u_z, \est - \theta \rangle\right| > t \right] \leq 2 \exp\left(-\frac{c t^2}{\cpsi^2}\right).
    \]
    The size of the set $\mathcal{Z}_{n,k}$ is bounded by $k^n$. Applying a union bound:
    \[
    \mathbb{P}\left[ \sup_{z \in \mathcal{Z}_{n,k}} \left|\langle u_z, \est - \theta \rangle\right| > t \right] \leq 2 k^n \exp\left(-\frac{c t^2}{\cpsi^2}\right) = 2 \exp\left(n \log k -\frac{c t^2}{\cpsi^2}\right).
    \]
    Setting $t = C \cpsi \sqrt{n \log k}$ for a sufficiently large constant $C$ ensures the probability is bounded by $\exp(-C' n \log k)$, concluding the proof.
\end{proof}

\begin{lemma}
    \label{lemma:concentration_3}
   For any constant \(C^{\prime}>0\), there exists a constant \(C>0\) only depending on \(C^{\prime}\), such that
    \[
    \left|\left\langle\frac{\hat{\theta}-\tilde{\theta}}{\|\hat{\theta}-\tilde{\theta}\|}, \est- \theta\right\rangle\right| \leq C \cpsi\sqrt{k^2 + n \log k}
    \]
    with probability at least \(1-\exp \left(-C^{\prime} n \log k\right)\).
\end{lemma}

\begin{proof}[Proof of\ \ \Cref{lemma:concentration_3}]
    We explicitly compute the inner product. Recall that $\hat{\theta}_{ij} = \bar{\est}_{ab}(\hat{z})$ and $\tilde{\theta}_{ij} = \bar{\theta}_{ab}(\hat{z})$ for $i \in \hat{z}^{-1}(a)$ and $j \in \hat{z}^{-1}(b)$. Therefore:
    \begin{align*}
        \langle \hat{\theta}-\tilde{\theta}, \est- \theta \rangle &= \sum_{a,b \in [k]} \sum_{i \in \hat{z}^{-1}(a)} \sum_{j \in \hat{z}^{-1}(b)} \langle \bar{\est}_{ab}(\hat{z}) - \bar{\theta}_{ab}(\hat{z}), \est_{ij} - \theta_{ij} \rangle_{\mathcal{X}} \\
        &= \sum_{a,b \in [k]} \left\langle \bar{\est}_{ab}(\hat{z}) - \bar{\theta}_{ab}(\hat{z}), \sum_{i \in \hat{z}^{-1}(a)} \sum_{j \in \hat{z}^{-1}(b)} (\est_{ij} - \theta_{ij}) \right\rangle_{\mathcal{X}} \\
        &= \sum_{a,b \in [k]} n_a n_b \left\langle \bar{\est}_{ab}(\hat{z}) - \bar{\theta}_{ab}(\hat{z}), \bar{\est}_{ab}(\hat{z}) - \bar{\theta}_{ab}(\hat{z}) \right\rangle_{\mathcal{X}} \\
        &= \sum_{a,b \in [k]} n_a n_b \|\bar{\est}_{ab}(\hat{z}) - \bar{\theta}_{ab}(\hat{z})\|_{\mathcal{X}}^2 \\
        &= \|\hat{\theta} - \tilde{\theta}\|^2.
    \end{align*}
    Thus, the term of interest simplifies to:
    \[
        \left|\left\langle\frac{\hat{\theta}-\tilde{\theta}}{\|\hat{\theta}-\tilde{\theta}\|}, \est- \theta\right\rangle\right| = \frac{\|\hat{\theta}-\tilde{\theta}\|^2}{\|\hat{\theta}-\tilde{\theta}\|} = \|\hat{\theta}-\tilde{\theta}\|.
    \]
    The result then follows from \Cref{lemma:concentration_1}.
\end{proof}

\subsection{Hölder-Smooth Model with Unbiased Edge-wise Estimators}
\label{sec:holder_convergence}

We now suppose that \(W\) is Hölder-smooth, meaning that there exists constants  \(\alpha, M>0\) such that that for any \((x,y),(u,v) \in [0,1]^2\), we have
\begin{equation}
    \label{eq:holder_smooth_def}
    \|W(x,y)-W(u,v)\|_\mathcal{X} \leq M\left(|x-u| + |y-v|\right)^\alpha.
\end{equation}
We use the notation \(W \in \mathcal{H}_{\alpha}(M)\) to denote that \(W\) is \(\alpha\)-Hölder smooth with constant \(M\). In this section, we suppose that for any \(i\neq j\), \(A_{ij}\) is sampled from a distribution parametrized by \(\theta_{ij}\), and that \(\theta_{ij} = W(\xi_i,\xi_j)\) for some iid latent variables \(\xi_i \sim  \UnifDist\). We will also assume that \(\Est\) satisfies \Cref{assumption:edge_process_est}.

\subsubsection{Oracle}

We first show the existence of a good approximation of \(\theta\) by an oracle \(\theta^*\):

\begin{lemma}
    \label{lemma:oracle_approx}
    If \(\ \theta_{ij} = W(\xi_i,\xi_j)\) for some \(W \in \mathcal{H}_{\alpha}(M)\), then there exists \(z^* \in \mathcal{Z}_{n,k}\), such that $$
\frac{1}{n^2} \sum_{a, b \in[k]} \sum_{\left\{i \neq j: z^*(i)=a, z^*(j)=b\right\}}\left\|\theta_{i j}-\bar{\theta}_{a b}\left(z^*\right)\right\|_\mathcal{X}^2 \leq C M^2\left(\frac{1}{k^2}\right)^{\alpha \wedge 1},
$$
for some universal constant \(C>0\).
\end{lemma}

\begin{proof}[Proof of \cref{lemma:oracle_approx}]
    The proof is the same as the proof of \citet[Lemma 2.1]{gao_rateoptimal_2015}. We define \(z^*:[n] \rightarrow[k]\) by
    \[
    \left(z^*\right)^{-1}(a)=\left\{i \in[n]: \xi_i \in U_a := \left[\frac{a-1}{k},\frac{a}{k}\right).\right\}
    \]
    We use the notation \(n_a^*=\left|\left(z^*\right)^{-1}(a)\right|\) for each \(a \in[k]\) and \(Z_{a b}^*=\left\{(u, v): z^*(u)=a, z^*(v)=b\right\}\) for \(a, b \in[k]\). By such construction of \(z^*\), for \(i, j\) such that \(\xi_i \in U_a, \xi_j \in U_b\) with \(a \neq b\), we have
    \[
    \begin{aligned}
    \left\|W\left(\xi_i, \xi_j\right)-\bar{\theta}_{a b}\left(z^*\right)\right\|_\mathcal{X} = & \left\|W\left(\xi_i, \xi_j\right)-\frac{1}{n_a^* n_b^*} \sum_{(u, v) \in Z_{a b}^*} W\left(\xi_u, \xi_v\right)\right\|_{\mathcal{X}} \\
    \leq & \frac{1}{n_a^* n_b^*} \sum_{(u, v) \in Z_{a b}^*}\left\|W\left(\xi_i, \xi_j\right)-W\left(\xi_u, \xi_v\right)\right\|_{\mathcal{X}} \\
    \leq & \frac{1}{n_a^* n_b^*} \sum_{(u, v) \in Z_{a b}^*} M\left(\left|\xi_i-\xi_u\right|+\left|\xi_j-\xi_v\right|\right)^{\alpha \wedge 1} \\
    \leq & C_1 M k^{-(\alpha \wedge 1)},
    \end{aligned}
    \]
    for some \(C_1>0\). The second inequality above due to \cref{eq:holder_smooth_def}. Similar results also hold for the case \(a=b\). Summing over \(i,j \in[n]\), the proof is complete.
\end{proof}

\begin{remark}
    The oracle approximation error depends only on the Hölder smoothness of $W$ and the number of blocks $k$, not on the nature of the decorations or the temporal sample size $T$. This is because the oracle $\theta^*$ is obtained by averaging the true parameters $\theta_{ij}$ within each block; a purely spatial operation on the latent space $[0,1]^2$. Its accuracy is governed entirely by how well a piecewise-constant function on $k^2$ blocks approximates $W$. In particular, increasing the observation window $T$ cannot reduce the oracle error, which highlights the fundamental role of the network size $n$ (through the choice of $k$) in controlling the nonparametric approximation.
\end{remark}

\subsubsection{Main result}
\label{subsec:chapter_B_abstract_result_conv}
\theoremAbstractResultConv*

The proof follows similarly to the block model case in \cref{subsec:block_model_conv_proof}, but we need to redefine \(z^*\) and \(Q^*\) as the one defining the oracle, implying a good approximation of \(\theta\) by an oracle \(\theta^*\). Using the similar argument as outlined in the beginning of \cref{subsec:block_model_conv_proof}, we get
\begin{equation}
    \label{eq:holder_UB_loss}
    \begin{aligned}
    \frac{1}{2}\|\hat{\theta}-\theta^*\|^2 \leq & \left\langle\hat{\theta}-\theta^*, \est- \theta^*\right\rangle \\
    = & \langle\hat{\theta}-\tilde{\theta}, \est- \theta\rangle+\left\langle\tilde{\theta}-\theta^*, \est- \theta\right\rangle+\left\langle\hat{\theta}-\theta^*, \theta-\theta^*\right\rangle \\
    \leq & \|\hat{\theta}-\tilde{\theta}\|\left|\left\langle\frac{\hat{\theta}-\tilde{\theta}}{\|\hat{\theta}-\tilde{\theta}\|}, \est- \theta\right\rangle\right|+\left(\|\tilde{\theta}-\hat{\theta}\|+\|\hat{\theta}-\theta^*\|\right)\left|\left\langle\frac{\tilde{\theta}-\theta^*}{\|\tilde{\theta}-\theta^*\|}, \est- \theta\right\rangle\right| \\
    &  + \|\hat{\theta}-\theta^*\|\|\theta-\theta^*\|.
    \end{aligned}
\end{equation}
All the terms above are very similar to the ones in the block model case, apart from \(\left|\left\langle\frac{\tilde{\theta}-\theta^*}{\|\tilde{\theta}-\theta^*\|}, \est- \theta\right\rangle\right|\).

\begin{lemma}
    \label{lemma:concentration_4}
    For any constant \(C^{\prime}>0\), there exists a constant \(C>0\) only depending on \(C^{\prime}\), such that
    \[
    \left|\left\langle\frac{\tilde{\theta}-\theta^*}{\left\|\tilde{\theta}-\theta^*\right\|}, \est- \theta\right\rangle\right| \leq C \cpsi\sqrt{n \log k}
    \]
    with probability at least \(1-\exp \left(-C^{\prime} n \log k\right)\).
\end{lemma}

\begin{proof}[Proof of\ \ \Cref{lemma:concentration_4}]
    This proof is identical to the proof of \Cref{lemma:concentration_2}.
\end{proof}

\begin{proof}[Proof of\ \ \Cref{theorem:abstract_result_conv} when \(\bias=0\)]
    To better organize \cref{eq:holder_UB_loss}, we follow \citet{gao_rateoptimal_2015} and introduce the notation
    \[
    \begin{gathered}
    L=\|\hat{\theta}-\theta^*\|, \quad R=\|\tilde{\theta}-\hat{\theta}\|, \quad B=\left\|\theta-\theta^*\right\| \\
    E=\left|\left\langle\frac{\hat{\theta}-\tilde{\theta}}{\|\hat{\theta}-\tilde{\theta}\|}, \est-\theta\right\rangle\right|, \quad F=\left|\left\langle\frac{\tilde{\theta}-\theta^*}{\|\tilde{\theta}-\theta^*\|}, \est-\theta\right\rangle\right|.
    \end{gathered}
    \]
    Then, from \cref{eq:holder_UB_loss}, we have

    \[
    L^2 \leq 2 R E+2(L+R) F+2 L B.
    \]
    By solving this quadratic inequality of \(L\), we can get
    \begin{equation}
        \label{eq:quadratic_inequality}
        L^2 \leq \max \left\{16(F+B)^2, 4 R(E+F)\right\}.
    \end{equation}

    By \Cref{lemma:concentration_1,lemma:concentration_2,lemma:concentration_3,lemma:concentration_4,lemma:oracle_approx}, for any constant $C^{\prime}>0$, there exist constants $C$ only depending on $C^{\prime}, M$, such that
    \begin{equation}
        \label{eq:UB_quadratic_terms}
        \begin{aligned}
            & B^2 \leq C n^2\left(\frac{1}{k^2}\right)^{\alpha \wedge 1}, \quad  \quad  && F^2 \leq C \cpsi^2n \log k, \\
            & R^2 \leq C\cpsi^2 \left(k^2+n \log k\right), && E^2 \leq C\cpsi^2\left(k^2+n \log k\right),
        \end{aligned}
    \end{equation}
    with probability at least $1-\exp \left(-C^{\prime} n\right)$. By \Cref{eq:quadratic_inequality}, we have
    \[L^2 \leq C_1\left(n^2\left(\frac{1}{k^2}\right)^{\alpha \wedge 1} + \cpsi^2\left(k^2+n \log k\right)\right)\]
    with probability at least \(1-\exp \left(-C^{\prime} n\log n\right)\) for some constant $C_1$. Hence, there is some constant $C_2$ such that
    \begin{equation}
        \begin{aligned}
            \frac{1}{n^2}\sum_{i,j \in [n]}\|\hat{\theta}_{ij}-\theta_{ij}\|_\mathcal{X}^2 &\leq \frac{2}{n^2}(L^2 + B^2) \\
            & \leq C_2\left(\left(\frac{1}{k^2}\right)^{\alpha \wedge 1} + \cpsi^2\left(\frac{k^2}{n^2} + \frac{\log k}{n}\right)\right)
        \end{aligned}
    \end{equation}
    with probability at least \(1-\exp \left(-C^{\prime} n\log n\right)\). When \(\alpha \geq 1\), picking \(k=\lceil\sqrt{n} \rceil\) gives the desired result. When \(\alpha < 1\), we pick \(k = \lceil n^{1/(\alpha+1)}\rceil\) to get the desired result.
\end{proof}

\subsection{Accounting for Bias in Edge-wise Estimators}
\label{subsec:bias_estimation}
We now deal with the case $\bias > 0$. Let $\Eest_{ij} = \bbE\left[\est_{ij}\right]$. We decompose the least-squares error (\cref{eq:mse_loss}) of our estimator as follows:
\begin{equation}
    \label{eq:impact_bias}
    \begin{aligned}
        \sum_{i,j=1}^{n}\left\|\hat{\theta}_{ij}-\theta_{ij}\right\|^2 &= \sum_{i,j=1}^{n}\left\|\hat{\theta}_{ij}-\Eest_{ij} + \Eest_{ij}-\theta_{ij}\right\|^2\\
        &\leq 2 \sum_{i,j=1}^{n}\left\|\hat{\theta}_{ij}-\Eest_{ij}\right\|^2 + 2 \sum_{i,j=1}^{n}\left\|\Eest_{ij}-\theta_{ij}\right\|^2\\
        & \leq 2 \sum_{i,j=1}^{n}\left\|\hat{\theta}_{ij}-\Eest_{ij}\right\|^2 + 2 n^2 \bias^2,
    \end{aligned}
\end{equation}
where we use the fact that \(\left\|\Eest_{ij}-\theta_{ij}\right\|_{\mathcal{X}} = \left\|\bbE\left[\est_{ij}\right]-\theta_{ij}\right\|_{\mathcal{X}} \leq \bias\) by \Cref{assumption:edge_process_est}, and thus we obtain
\begin{equation}
    \label{eq:upper_bound_loss_bias}
    \frac{1}{n^2}\sum_{i,j\in[n]}\|\hat{\theta}_{ij}-\theta_{ij}\|^2 \leq \frac{2}{n^2}\sum_{i,j=1}^{n}\left\|\hat{\theta}_{ij}-\Eest_{ij}\right\|^2 + 2\bias^2.
\end{equation}

The results from the block model case (\Cref{subsec:block_model_conv_proof}) hold by substituting $(\Eest_{ij})$ for $(\theta_{ij})$.

\begin{proof}[Proof of\ \ \Cref{theorem:abstract_result_blockmodel} when $\bias \neq 0$]
    From \Cref{eq:upper_bound_loss_bias}, we are left with bounding
    $$\frac{1}{n^2}\sum_{i,j=1}^{n}\left\|\hat{\theta}_{ij}-\Eest_{ij}\right\|^2.$$
    This can be done by following the same steps as in the proof of \Cref{theorem:abstract_result_blockmodel} when $\bias = 0$, but replacing $\theta$ with $\Eest$ and $\tilde{\theta}$ with $\tilde{\Eest}$ accordingly. The proofs of the corresponding versions of \Cref{lemma:concentration_1,lemma:concentration_2,lemma:concentration_3} follow mutatis mutandis.

    We then get that for any constant \(C^{\prime}>0\), there exists a constant \(C>0\) only depending on \(C^{\prime}\), such that
    \begin{equation}
        \begin{aligned}
            \frac{1}{n^2}\sum_{i,j\in[n]}\|\hat{\theta}_{ij}-\theta_{ij}\|^2 &\leq \frac{2}{n^2}\sum_{i,j=1}^{n}\left\|\hat{\theta}_{ij}-\Eest_{ij}\right\|^2 + 2\bias^2 \\
            &\leq C\left(\cpsi^2\left(\frac{k^2}{n^2}+\frac{\log k}{n}\right) + \bias^2\right)
        \end{aligned}
    \end{equation}
    with probability at least \(1-\exp \left(-C^{\prime} n\log k\right)\), concluding the proof.
\end{proof}

For the continuous case, we proceed similarly, but we need to make sure that we can get a good oracle approximation of $(\Eest_{ij})$ by a block model, and track how this bias is propagated trough \Cref{eq:holder_UB_loss}.

\begin{proof}[Proof of\ \ \Cref{theorem:abstract_result_conv} when $\bias \neq 0$]
   We proceed similarly to the case when $\bias = 0$. Replacing $\theta$ with $\Eest$ in \Cref{eq:holder_UB_loss}, we are left with
   \[ L=\|\hat{\Eest}-\Eest^*\|, \quad \text{and} \; B=\left\|\Eest-\Eest^*\right\|,\]
    and we now have for any constant $C^{\prime}>0$, there exist constants $C$ only depending on $C^{\prime}, M$, such that
    \begin{equation}
        \begin{aligned}
            \frac{1}{n^2}\sum_{i,j \in [n]}\|\hat{\theta}_{ij}-\theta_{ij}\|_\mathcal{X}^2 &\leq \frac{2}{n^2} \sum_{i,j=1}^{n}\left\|\hat{\theta}_{ij}-\Eest_{ij}\right\|^2 + 2 \bias^2  \\
            & \leq \frac{4}{n^2}(L^2 + B^2) + 2\bias^2,
        \end{aligned}
    \end{equation}
    with probability at least \(1-\exp \left(-C^{\prime} n\log n\right)\). The difference with the case $\bias = 0$ is that $B^2$ is now upper bounded by $Cn^2\left(\left(\frac{1}{k^2}\right)^{\alpha \wedge 1}+\bias^2\right)$, and thus the bias term $\bias^2$ will be added to the final rate of convergence; more precisely, there is some constant $C_2$ such that

    \begin{equation*}
        \begin{aligned}
            \frac{1}{n^2}\sum_{i,j \in [n]}\|\hat{\theta}_{ij}-\theta_{ij}\|_\mathcal{X}^2 &\leq C_2\left(\left(\frac{1}{k^2}\right)^{\alpha \wedge 1} + \cpsi^2\left(\frac{k^2}{n^2} + \frac{\log k}{n}\right) + \bias^2\right) + 2\bias^2\\
            & \leq (C_2 + 2)\left(\left(\frac{1}{k^2}\right)^{\alpha \wedge 1} + \cpsi^2\left(\frac{k^2}{n^2} + \frac{\log k}{n}\right) + \bias^2\right)
        \end{aligned}
    \end{equation*}
    with probability at least \(1-\exp \left(-C^{\prime} n\log n\right)\), concluding the proof.
\end{proof}

\begin{lemma}[Adaptation of \Cref{lemma:oracle_approx} when $\bias \neq 0$]
    \label{lemma:oracle_approx_bias}
    Let \(\ \theta_{ij} = W(\xi_i,\xi_j)\;\) for some \(W \in \mathcal{H}_{\alpha}(M)\), and we define $\Eest_{ij} = \bbE\left[\Est(A_{ij})\right]$. Then there exists \(z^* \in \mathcal{Z}_{n,k}\), such that $$
    \frac{1}{n^2} \sum_{a, b \in[k]} \sum_{\left\{i \neq j: z^*(i)=a, z^*(j)=b\right\}}\left\|\Eest_{i j}-\bar{\Eest}_{a b}\left(z^*\right)\right\|_\mathcal{X}^2 \leq C \left(M^2\left(\frac{1}{k^2}\right)^{\alpha \wedge 1}+\bias^2\right),
$$
for some universal constant \(C>0\).
\end{lemma}

\begin{proof}[Proof of\ \ \Cref{lemma:oracle_approx_bias}]
    Using the same construction of \(z^*\) as in the proof of \Cref{lemma:oracle_approx}, we have for \(i, j\) such that \(\xi_i \in U_a, \xi_j \in U_b\) with \(a \neq b\),
    \begin{equation*}
        \begin{aligned}
            \left\|\Eest_{ij}-\bar{\Eest}_{a b}\left(z^*\right)\right\|_\mathcal{X} = & \left\|\Eest_{ij}-\frac{1}{n_a^* n_b^*} \sum_{(u, v) \in Z_{a b}^*}\Eest_{uv}\right\|_{\mathcal{X}} \\
            \leq & \frac{1}{n_a^* n_b^*} \sum_{(u, v) \in Z_{a b}^*}\left\|\Eest_{ij} - \Eest_{uv}\right\|_{\mathcal{X}} \\
            \leq &  \frac{1}{n_a^* n_b^*} \sum_{(u, v) \in Z_{a b}^*}\left\|\Eest_{ij} - \theta_{ij} + \theta_{ij} - \theta_{uv} + \theta_{uv} - \Eest_{uv}\right\|_{\mathcal{X}} \\
            \leq &  \frac{1}{n_a^* n_b^*} \sum_{(u, v) \in Z_{a b}^*}\left(\left\|\Eest_{ij} - \theta_{ij}\right\|_{\mathcal{X}} + \left\|\theta_{ij} - \theta_{uv}\right\|_{\mathcal{X}} + \left\|\theta_{uv} - \Eest_{uv}\right\|_{\mathcal{X}}\right) \\
            \leq & \frac{1}{n_a^* n_b^*} \sum_{(u, v) \in Z_{a b}^*} M\left(\left|\xi_i-\xi_u\right|+\left|\xi_j-\xi_v\right|\right)^{\alpha \wedge 1} + 2\bias \\
            \leq & C_1 \left(k^{-(\alpha \wedge 1)} + \bias\right),
        \end{aligned}
    \end{equation*}
    for some \(C_1>0\). Squaring and summing over \(i,j \in[n]\), the proof is complete.
\end{proof}

\section{Proof of Examples in \texorpdfstring{\Cref{sec:a_examples}}{Section 2.4}}
\label{sec:appendix-examples}

\subsection{Markov Chain Case}
\label{subsec:concentration_MC}

\lemmaConcentrationMC*

\newcommand{\psimcsquare}{32 \frac{d\sqrt{\tau_{\min }}}{T \pi_*^{3 / 2}} \log \left(\frac{2 d}{\pi_*}\right)}

\begin{proof}[Proof of\ \ \Cref{lemma:concentration_MC}]

    We use the Frobenius norm $\|\cdot\|_F$ to compare $\hat{P}$ and $P$. We proceed in two steps: first we bound the bias and then derive the sub-Gaussian concentration of $\|\hat{P}-\bbE[\hat{P}]\|_{F}$.

    \paragraph{Step 1: Bias bound.}
    We consider $\bias=\left\|\bbE\left[\hat{P}\right]-P\right\|_{F}$. Using \Cref{lemma:bias_empirical_transition}, we have
    \begin{equation}
        \label{eq:bias_bound_MC}
        \begin{aligned}
        \bias^2 &= \sum_{u,v \in [d]} |\bbE[\hat{P}_{uv}] - P_{uv}|^2 \\
        & \leq \sum_{u,v \in [d]} \left(\frac{\tau_{\min}}{(T-1)\pi_*^2} + \sqrt{\frac{2}{\pi_*}}\exp\left(-\frac{T\pi_*^2}{\tau_{\min}}\right)\right)^2\\
        &\leq 2d^2 \left(\frac{\tau_{\min}^2}{(T-1)^2\pi_*^4} + \frac{2}{\pi_*}\exp\left(-\frac{2T\pi_*^2}{\tau_{\min}}\right)\right).
    \end{aligned}
    \end{equation}

    Using \Cref{lemma:bias_lambert_function_proof} and the lower bound for $T$ in \Cref{lemma:concentration_MC}, we can further simplify the bias bound as follows:
    \begin{equation}
        \label{eq:bias_bound_MC_simplified}
        T > \frac{32\,d\,\tau_{\min}^{3/2}}{\pi_*^{5/2}\,\log(2d/\pi_*)} \geq \frac{2\tau_{\min}}{\pi_*^2}\log\left(\frac{2}{\pi_*}\right) \implies \bias^2 \leq \frac{4d^2\tau_{\min}^2}{(T-1)^2\pi_*^4}.
    \end{equation}

    \paragraph{Step 2: Sub-Gaussian concentration.}
    We use \Cref{lemma:concentration_transition_matrix} to get that for $\varepsilon > 0$,
    \begin{equation*}
        \mathbb{P}\left[\|\hat{P}-P\|_{F} > \varepsilon\right]  \leq  \exp \left(-\varepsilon^2/\tilde{\psi}^2\right),
    \end{equation*}
    where $\tilde{\psi}^2 = \psimcsquare$. Since $\|P - \bbE[\hat{P}]\|_{F}$ is deterministic, its sub-Gaussian norm satisfies $\|P - \bbE[\hat{P}]\|_{\psi_2} = \|P - \bbE[\hat{P}]\|_{F} / \sqrt{\log 2}$  (see \citet[Remark 2.5.7]{vershynin_highdimensional_2018}). Applying the triangle inequality for sub-Gaussian norms gives
    \begin{equation*}
        \|\hat{P}-\bbE\left[\hat{P}\right]\|_{\psi_2} \leq \|\hat{P}-P\|_{\psi_2} + \frac{\left\|P-\bbE\left[\hat{P}\right]\right\|_{F}}{\sqrt{\log 2}} \leq \tilde{\psi} + \frac{\bias}{\sqrt{\log 2}}.
    \end{equation*}

    It remains to show that $\bias/\sqrt{\log 2} \leq \tilde{\psi}$.  By \cref{eq:bias_bound_MC_simplified}, this reduces to
    \begin{equation}
        \label{eq:bias_absorb_squared}
        \frac{4d^2\tau_{\min}^2}{(T-1)^2\pi_*^4\log2} \leq \frac{32\,d\sqrt{\tau_{\min}}}{T\pi_*^{3/2}}\log\!\left(\frac{2d}{\pi_*}\right).
    \end{equation}
    Rearranging \cref{eq:bias_absorb_squared} (all factors are positive) and using $T/(T-1)^2 \leq 4/T$ for $T \geq 2$, it is enough to show
    \[
        \frac{16\,d\,\tau_{\min}^{3/2}}{\pi_*^{5/2}\,\log 2} \cdot \frac{1}{T}
        \leq \log\!\left(\frac{2d}{\pi_*}\right),
    \]
    i.e.\
    \[
        T \geq \frac{16\,d\,\tau_{\min}^{3/2}}{\pi_*^{5/2}\,\log(2)\log(2d/\pi_*)},
    \]
    which is a consequence of the lower bound for $T$ in \Cref{lemma:concentration_MC}.  We thus obtain
    \begin{equation*}
        \|\hat{P}-\bbE\left[\hat{P}\right]\|_{\psi_2} \leq 2\tilde{\psi}.
    \end{equation*}

    Combining the two steps, we obtain $\cpsi^2 \leq C\tilde{\psi}^2$ and $\bias^2 \leq C' d^2 \pi_*^{-1}\exp(-2T\pi_*^2/\tau_{\min})$, which concludes the proof.
\end{proof}

\begin{lemma}
    \label{lemma:bias_empirical_transition}
     Using the notation of \Cref{lemma:concentration_MC}, for any $u,v \in [d]$, we have
    \begin{equation*}
        \bias_{uv} = \left|\bbE\left[\hat{P}_{uv}\right] - P_{uv}\right| \leq \frac{\tau_{\min}}{(T-1)\pi_*^2} + \sqrt{\frac{2}{\pi_*}}\exp\left(-\frac{T\pi_*^2}{\tau_{\min}}\right)=O\left(\frac{\tau_{\min}}{T\pi_*^2}\right).
    \end{equation*}
\end{lemma}

\begin{proof}[Proof of\ \ \Cref{lemma:bias_empirical_transition}]

Let $N_u = \sum_{t=1}^{T-1} \indicator{X_t = u}$ denote the number of visits to state $u$, and let $N_{uv} = \sum_{t=1}^{T-1} \indicator{X_t = u, X_{t+1} = v}$ be the number of observed transitions from $u$ to $v$. Let $\nu = \mathbb{E}[N_u] = (T-1)\pi_u$. We consider
\begin{equation*}
    \hat{P}_{uv} = \frac{N_{uv}}{N_u} \indicator{N_u > 0} + \frac{1}{d} \indicator{N_u = 0}.
\end{equation*}

Define $M = N_{uv} - P_{uv}N_u = \sum_{t=1}^{T-1} \indicator{X_t = u}\big(\indicator{X_{t+1} = v} - P_{uv}\big)$. The sequence of summands forms a martingale difference sequence with respect to the filtration $\mathcal{F}_t = \sigma(X_1, \dots, X_t)$, implying $\mathbb{E}[M] = 0$. Since $M = 0$ whenever $N_u = 0$, it also holds that $\mathbb{E}[M \indicator{N_u > 0}] = 0$.

The estimation error can be decomposed as:
\begin{equation*}
    \hat{P}_{uv} - P_{uv} = \frac{M}{N_u} \indicator{N_u > 0} + \left(\frac{1}{d} - P_{uv}\right) \indicator{N_u = 0}.
\end{equation*}

To compute the expectation of the ratio without conditioning on $N_u$, we use that for $x > 0$:
\begin{equation*}
    \frac{1}{x} = \frac{1}{\nu} - \frac{x - \nu}{\nu^2} + \frac{(x - \nu)^2}{\nu^2 x}.
\end{equation*}
Substituting $x = N_u$ and multiplying by $M \indicator{N_u > 0}$ yields:
\begin{equation*}
    \frac{M}{N_u} \indicator{N_u > 0} = \frac{M}{\nu} \indicator{N_u > 0} - \frac{M(N_u - \nu)}{\nu^2} \indicator{N_u > 0} + \frac{M(N_u - \nu)^2}{\nu^2 N_u} \indicator{N_u > 0}.
\end{equation*}

Taking the expectation of both sides, and noting that $M(N_u - \nu) = 0$ when $N_u = 0$, the first term vanishes and we obtain:
\begin{equation}
    \label{eq:bias_decomp}
    \mathbb{E}\left[\frac{M}{N_u} \indicator{N_u > 0}\right] = -\frac{\text{Cov}(M, N_u)}{\nu^2} + \mathbb{E}\left[ \frac{M}{N_u} \frac{(N_u - \nu)^2}{\nu^2} \indicator{N_u > 0} \right].
\end{equation}

We bound the two terms in \cref{eq:bias_decomp} separately. For the covariance term, we apply the Cauchy-Schwarz inequality.
Using properties of martingale difference sequences, the variance of $M$ is bounded by:
\begin{equation*}
    \text{Var}(M) = \sum_{t=1}^{T-1} \mathbb{E}\left[ \indicator{X_t = u} (\indicator{X_{t+1=v}}-P_{uv})^2 \right] \leq \frac{1}{4}\mathbb{E}[N_u] = \frac{\nu}{4}.
\end{equation*}
Using \citet[Corollary 2.10]{paulin_concentration_2015} and properties of sub-Gaussian random variables, we get that the variance of the occupation time satisfies $\text{Var}(N_u) \leq \tau_{\min} (T-1)$, where $\tau_{\min}$ is the mixing time. Thus:
\begin{equation*}
    \frac{|\text{Cov}(M, N_u)|}{\nu^2} \leq \frac{\sqrt{(\nu/4)\tau_{\min} (T-1)}}{\nu^2} = \frac{\sqrt{\pi_u\tau_{\min}}}{2\nu\pi_u}.
\end{equation*}

For the remainder term, note that conditional on $N_u > 0$, we have $\left| \frac{M}{N_u} \right| = \left| \frac{N_{uv}}{N_u}-P_{uv} \right| \leq 1$. Therefore:
\begin{equation*}
    \left| \mathbb{E}\left[ \frac{M}{N_u} \frac{(N_u - \nu)^2}{\nu^2} \indicator{N_u > 0} \right] \right| \leq \frac{\mathbb{E}[(N_u - \nu)^2]}{\nu^2} = \frac{\text{Var}(N_u)}{\nu^2} \leq \frac{\tau_{\min}}{\nu\pi_u}.
\end{equation*}

Combining these bounds, the total bias $\bias_{uv}$ is bounded by:
\begin{equation*}
    |\bias_{uv}| \leq \frac{\tau_{\min}}{\nu\pi_u} + \left| \frac{1}{d} - P_{uv} \right| \mathbb{P}[N_u = 0] \leq \frac{\tau_{\min}}{\pi_*^2(T-1)}  + \mathbb{P}[N_u = 0].
\end{equation*}

Using \citet[Corollary 2.10 and Proposition 3.10]{paulin_concentration_2015},
    \begin{equation*}
        \bbP\left[N_u = 0\right] \leq \bbP\left[\left|N_u - T\pi_u\right| > T\pi_u\right] \leq \sqrt{\frac{2}{\pi_*}}\exp\left(-\frac{T\pi_u^2}{\tau_{\min}}\right) \leq \sqrt{\frac{2}{\pi_*}}\exp\left(-\frac{T\pi_*^2}{\tau_{\min}}\right),
    \end{equation*}
so $\mathbb{P}[N_u = 0]$ decays exponentially with $T$, and the bias $\bias_{uv}$  is $O\left(\frac{\tau_{\min}}{T\pi_u^2}\right)$.
\end{proof}

\begin{lemma}
    \label{lemma:concentration_transition_matrix}
    Using the notation of \Cref{lemma:concentration_MC}, for any $\varepsilon \in [0,d]$ we have
    \begin{equation*}
         \mathbb{P}\left[\|\hat{P}-P\|_{F} > \varepsilon\right]  \leq  \ \frac{4d^2}{\sqrt{\pi_*}}  \cdot \exp \left(-\frac{\varepsilon^2 T\pi_*^{3/2}}{16d\sqrt{\tau_{\min}}}\right),
    \end{equation*}
    and thus is a sub-Gaussian random variable with sub-Gaussian norm $$\|\hat{P}-P\|_{\psi_2} \leq \psimcsquare.$$
\end{lemma}

\begin{proof}[Proof of\ \ \Cref{lemma:concentration_transition_matrix}]
    We modify the end of the proofs in \citet{wolfer_minimax_2019} by carefully tuning the number of times each state is visited. To do so, we need to change the proofs of \citet[Theorem 1 and Lemma 6]{wolfer_minimax_2019} as described below. We note that as in the original paper \citep{wolfer_minimax_2019}, the starting distribution of the chain does not need to be the stationary distribution for our results to hold. We restate the result we want to prove for the sake of clarity.
    \begin{equation*}
    \begin{aligned}
        \mathbb{P}\left[\|\hat{P}-P\|_{F}>\varepsilon \wedge \forall i: N_i \in\left[l_i, u_i\right]\right] & \leq \sum_i \mathbb{P}\left[\left\|\sum_{t=1}^T \boldsymbol{Y}_t\right\|_{F}>\frac{\varepsilon}{\sqrt{2 d}} l_i \wedge \max \left\{\left\|\boldsymbol{W}_{\text{row}, T}\right\|_{F},\left\|\boldsymbol{W}_{\text{col}, T}\right\|_{F}\right\} \leq u_i\right] \\
        & \leq \sum_i(d+1) \cdot \exp \left(-\frac{\varepsilon^2 l_i^2}{2 d\left(u_i+\varepsilon l_i/(3 \sqrt{2 d})\right)}\right)\\
        & \leq \sum_i 2 d  \cdot \exp \left(-\frac{\varepsilon^2 l_i^2}{4 d u_i}\right),\\
    \end{aligned}
    \end{equation*}
    where the last line is obtained by realizing that \(\varepsilon \leq \sqrt{d}\), otherwise the bound is trivial since \(\|\hat{P}-P\|_{F} \leq \sqrt{d}\). Furthermore, if we pick \([l_i,u_i] = \left[T\pi_i(1-\delta),  T \pi_i (1+\delta) \right]\) for some \(\delta \in [0,1]\), we get
    \begin{equation*}
        \mathbb{P}\left[\|\hat{P}-P\|_{F}>\varepsilon \wedge \forall i: N_i \in\left[T\pi_i(1-\delta),  T \pi_i (1+\delta) \right]\right] \leq 2 d^2  \cdot \exp \left(-\frac{\varepsilon^2 T\pi_* (1-\delta)^2}{4 d (1+\delta)}\right).
    \end{equation*}

    We now want to bound the probability that $\exists i$ such that $N_i \notin \left[(1-\delta) T \pi_i,(1+\delta)T\pi_i\right]$. Using \citet[Corollary 2.10 and Proposition 3.10]{paulin_concentration_2015}
    \begin{equation*}
        \mathbb{P}\left[N_i \notin  \left[(1-\delta) T \pi_i,(1+\delta)T\pi_i\right]\right] = \mathbb{P}\left[\left|N_i - T\pi_i\right|>\delta T \pi_i \right] \leq \sqrt{\frac{2}{\pi_*}}\exp\left(-\frac{\delta^2T\pi_i^2}{\tau_{\min}}\right),
    \end{equation*}
    where $\tau_{\min} = \inf\limits_{0 \leq \epsilon \leq 1}\tau(\epsilon)\left(\frac{2-\epsilon}{1-\epsilon}\right)^2$ and $\tau(\epsilon)$ is the mixing time of the Markov chain \citep{paulin_concentration_2015}.

    With a union bound argument, we get
    \begin{equation*}
        \mathbb{P}\left[\exists i \stext{s.t.} N_i \notin  \left[(1-\delta) T \pi_i,(1+\delta)T\pi_i\right]\right] \leq d\sqrt{\frac{2}{\pi_*}}\exp\left(-\frac{\delta^2T\pi_*^2}{\tau_{\min}}\right).
    \end{equation*}
    Combining the two results, we get that
    \begin{equation}
        \label{eq:bound_on_bad_event_with_main}
        \begin{aligned}
            \mathbb{P}\left[\|\hat{P}-P\|_{F} > \varepsilon\right] & \leq   \mathbb{P}\left[\|\hat{P}-P\|_{F}>\varepsilon \wedge \forall i: N_i \in\left[l_i, u_i\right]\right] +  \mathbb{P}\left[\exists i \stext{s.t.} N_i \notin  \left[(1-\delta) T \pi_i,(1+\delta)T\pi_i\right]\right]\\
            & \leq 2 d^2  \cdot \exp \left(-\frac{\varepsilon^2 T\pi_* (1-\delta)^2}{4 d (1+\delta)}\right) + d\sqrt{\frac{2}{\pi_*}}\exp\left(-\frac{\delta^2T\pi_*^2}{\tau_{\min}}\right) \\
            & \leq \frac{2d}{\sqrt{\pi_*}}\left(\exp \left(-\frac{\varepsilon^2 T\pi_* (1-\delta)^2}{4 d (1+\delta)} + \log(d)\right) + \exp\left(-\frac{\delta^2T\pi_*^2}{\tau_{\min}}\right)\right). \\
        \end{aligned}
    \end{equation}

    Using \Cref{lemma:delta_min}, we can choose $\delta = \delta_{\min}$ such that the first term in the last line of \cref{eq:bound_on_bad_event_with_main} dominates the second term. Using the fact that $(1-\delta)^2/(1+\delta)\geq (1-\delta)^2/2\geq (1-2\delta)/2$, and that $(1-2\delta_{\min})/2\geq \sqrt{\pi_*}/\left(2 + \sqrt{\tau_{\min}}\right)$, we get
    \begin{equation*}
        \mathbb{P}\left[\|\hat{P}-P\|_{F} > \varepsilon\right]  \leq  \frac{4d^2}{\sqrt{\pi_*}}  \cdot \exp \left(-\frac{\varepsilon^2 T\pi_*^{3/2}}{16d\sqrt{\tau_{\min}}}\right).
    \end{equation*}
    From \citet[Proposition 2.5.2]{vershynin_highdimensional_2018}, we have that \(\|P-\hat{P}\|_{F}\) is a sub-Gaussian random variable with sub-Gaussian norm
    \begin{equation*}
        \begin{aligned}
            \|P-\hat{P}\|_{\psi^2}^2 &\leq \frac{16d\sqrt{\tau_{\min}}}{\pi_*^{3/2}T}\log\left(\frac{4d^2}{\sqrt{\pi_*}}\right)\\
            & \leq \frac{16 d \sqrt{\tau_{\min }}}{\pi_*^{3 / 2} T} \cdot 2 \log \left(\frac{2 d}{\pi_*}\right) \\
            & \leq \psimcsquare.
        \end{aligned}
    \end{equation*}
\end{proof}

\begin{remark}
    There is an implicit assumption in \citet{wolfer_minimax_2019} that each state is visited at least once that we keep in our proof, and that motivates the lower bound on $T$. However, for practical application one needs to have a way to manage what happens if one of the chain does not visit a state, for instance by using a default value for the MLE in that case as we did in the proof of the bias. One could also consider a regularized (add-$\alpha$) estimator \citep{hao_learning_2018}  $\hat{P}_{uv}^{(\alpha)} = (N_{uv}+\alpha)/(N_u+d\alpha)$, which introduces a controlled bias of order $O(\alpha/(T\pi_*))$ but regularizes the denominator. Since the MLE is already conditionally unbiased, we conjecture that the improvement by this estimator would only be marginal in our setting.
\end{remark}

\subsubsection{Algebraic inequalities}

\begin{lemma}
\label{lemma:bias_lambert_function_proof}
Let $\pi_* \in (0,1),\tau_{\min} > 0$, and $T>1$. Then the following holds
\begin{equation*}
   T \geq \frac{2\tau_{\min}}{\pi_*^2}\log\left(\frac{2}{\pi_*}\right) \Rightarrow \frac{2}{\pi_*}\exp\left(-\frac{2T\pi_*^2}{\tau_{\min}}\right) \leq \frac{\tau_{\min}^2}{(T-1)^2\pi_*^4}.
\end{equation*}
\end{lemma}

\begin{proof}
Taking square roots of the target inequality (both sides positive) and then logarithms, the conclusion is equivalent to
\[
  \frac{1}{2}\log\!\frac{2}{\pi_*} + \log\!\bigl(s(T-1)\bigr) \leq sT,
\]
where $s = \pi_*^2/\tau_{\min} > 0$. Setting $L := \tfrac{1}{2}\log(2/\pi_*) \geq 0$ and $u = sT$, and using $s(T-1) < sT = u$, it suffices to show
\[
  L + \log u \leq u.
\]
The hypothesis gives $u = sT \geq 2L$. By the standard inequality $\log u \leq u/e$ (sharp at $u = e$),
\[
  L + \log u \;\leq\; L + \frac{u}{e}.
\]
Since $u \geq 2L$ and $1 - 1/e > 1/2$, we have $L \leq u/2 \leq u(1-1/e)$, so
\[
  L + \frac{u}{e} \;\leq\; u\!\left(1 - \frac{1}{e}\right) + \frac{u}{e} = u,
\]
which completes the proof.
\end{proof}

\begin{lemma}
    \label{lemma:delta_min}
    Given the notation of \Cref{lemma:concentration_MC}, consider $1 > \delta \geq \delta_{\min}$, where
    \begin{equation}
        \begin{aligned}
            \delta_{\min} &= \frac{\sqrt{\tau_{\min}}}{2\sqrt{\pi_*}+\sqrt{\tau_{\min}}}.
        \end{aligned}
    \end{equation}
    Then for all $\varepsilon \in (0,\sqrt{d})$ (since $\|\hat{P}-P\|_{F} \leq \sqrt{d}$ as they both are row stochastic matrices),
    \begin{equation*}
        \frac{\varepsilon^2 T\pi_* (1-\delta)^2}{4 d (1+\delta)} - \log(d)\leq \frac{\delta^2T\pi_*^2}{\tau_{\min}}.
    \end{equation*}
\end{lemma}

\begin{proof}[Proof of\ \ \Cref{lemma:delta_min}]
    We rearrange the desired inequality:
    \begin{align*}
        &  \frac{\varepsilon^2 T\pi_* (1-\delta)^2}{4 d (1+\delta)} - \log(d) \leq \frac{\delta^2T\pi_*^2}{\tau_{\min}} \\
        \Leftrightarrow \quad & \frac{\varepsilon^2 T\pi_* (1-\delta)^2}{4 d (1+\delta)} \leq \frac{\delta^2T\pi_*^2}{\tau_{\min}}  +  \log(d) \\
        \Leftarrow \quad & \frac{\varepsilon^2 T\pi_* (1-\delta)^2}{4 d (1+\delta)} \leq \frac{\delta^2T\pi_*^2}{\tau_{\min}} \\
        \Leftrightarrow \quad & \frac{\varepsilon^2 \tau_{\min}}{4 d \pi_*} \leq \delta^2\frac{1+\delta}{(1-\delta)^2}.
    \end{align*}

    Since $\varepsilon \in (0,\sqrt{d})$, we have $\varepsilon^2 < d$. Thus, it suffices to ensure
    \begin{equation}
        \label{eq:sufficient_condition}
        \frac{\tau_{\min}}{4\pi_*} \leq \delta^2\frac{1+\delta}{(1-\delta)^2}.
    \end{equation}

    Now, observe that for $\delta \in (0,1)$, we have $\frac{1+\delta}{(1-\delta)^2} > \frac{1}{(1-\delta)^2}$, so
    \begin{equation*}
        \label{eq:simpler_condition}
        \delta^2\frac{1+\delta}{(1-\delta)^2} > \frac{\delta^2}{(1-\delta)^2}. \quad \text{Therefore, \Cref{eq:sufficient_condition} holds if} \quad \frac{\delta^2}{(1-\delta)^2} \geq \frac{\tau_{\min}}{4\pi_*}.
    \end{equation*}

    Taking square roots and rearranging:
    \begin{align*}
        \frac{\delta}{1-\delta} &\geq \frac{\sqrt{\tau_{\min}}}{2\sqrt{\pi_*}} \\
        \Leftrightarrow \quad 2\sqrt{\pi_*} \cdot \delta &\geq \sqrt{\tau_{\min}}(1-\delta) \\
        \Leftrightarrow \quad \delta(2\sqrt{\pi_*} + \sqrt{\tau_{\min}}) &\geq \sqrt{\tau_{\min}} \\
        \Leftrightarrow \quad \delta &\geq \frac{\sqrt{\tau_{\min}}}{2\sqrt{\pi_*}+\sqrt{\tau_{\min}}} = \delta_{\min}.
    \end{align*}

    This shows that for $\delta \geq \delta_{\min}$, inequality \eqref{eq:simpler_condition} holds, which implies the desired result.
\end{proof}

\begin{remark}
    \label{remark:mc_improvements}
    We could consider a pooled estimator that, after an initial clustering step, estimates a single transition matrix from all edge chains within each block. A key difficulty is that, even within a correctly identified block $(a,b)$, different edges $(i,j)$ and $(u,v)$ with $z(i)=z(u)=a$, $z(j)=z(v)=b$ can be governed by \emph{different} transition kernels (since $W(\xi_i,\xi_j) \neq W(\xi_u,\xi_v)$ in general) so the pooled chains are independent but not necessarily identically distributed. \citet{leskela_robust_2026} recently studied the related problem of estimating a Markov chain transition matrix from multiple independent trajectories with heterogeneous kernels, and their analysis would be a natural starting point for studying such a pooled estimator in our setting. However, this approach breaks the two-stage decoupling that underlies our proof architecture and would require a fundamentally different analysis.
\end{remark}

\section{Simulation Study Details}
\label{sec:simulation_details}

For all simulation, we picked $k =\sqrt{n}$, and the latent variables $\xi_i$ were sampled uniformly from $[0,1]$. We use a greedy algorithm to optimize the node assignments, similar to the one used in \citet{olhede_network_2014}: it is a variant of the Tabu search algorithm \citep{glover_tabu_1997}.

\subsection{Block Model Case}
In the BALARM simulation, we generate a dynamic network with three latent blocks, each governed by a second-order Markov process with block-specific periodic activation patterns. The network is observed over a 10-day period at 15-minute resolution, yielding $T = 960$ time points. Each edge process is governed by a combination of autoregressive memory, a time-varying harmonic signal, and a block-specific bias. Specifically, the activation probability at time $t$ is given by
\[
\mathbb{P}[A_{ij}^{(t)} = 1 \mid A_{ij}^{(t-1)}, A_{ij}^{(t-2)}] = \text{logistic}(\phi_{ab}(t) + b_1 A_{ij}^{(t-1)} + b_2 A_{ij}^{(t-2)} + c_{ab}),
\]
where $(a,b)$ denotes the block pair, and the harmonic component $\phi_{ab}(t)$ is defined as follows (with $P = 96$):
\[
\phi_{ab}(t) =
\begin{cases}
-0.15 \cos(2\pi t/P) + 1.2 \sin(2\pi t/P) & \text{if } (a,b) = (1,1) \\
0 & \text{if } (a,b) \in \{(1,2),(2,2)\}\\
0.3 \cos(4\pi t/P) - 0.09 \sin(4\pi t/P) & \text{if } (a,b) = (1,3) \\
-0.8 \cos(2\pi t/P) + 0.2 \sin(2\pi t/P) & \text{if } (a,b) = (2,3) \\
-0.8 \cos(2\pi t/P) - 2.0 \sin(2\pi t/P) & \text{if } (a,b) = (3,3)
\end{cases}
\]
and the autoregressive coefficients $(b_1, b_2)$ and biases $c_{ab}$ for each block are:
\[
(b_1, b_2, c_{ab}) =
\begin{cases}
(0.15, 0.09, 0.2) & \text{for } (1,1) \\
(0.15, -0.15, -1.0) & \text{for } (1,2) \\
(0.15, -0.3, -1.0) & \text{for } (1,3) \\
(0.0, 0.0, -1.1) & \text{for } (2,2) \\
(-0.5, 0.0, -0.55) & \text{for } (2,3) \\
(-0.5, 0.5, -0.7) & \text{for } (3,3).
\end{cases}
\]

Comparing the edge processes for $(1,2)$ and $(2,2)$ highlights the need for methods that can handle explicit temporal dependence rather than just conditional independence on the temporal mean.

\subsection{Smooth Graphon Case}

In the smooth graphon case, the latent function $W(x, y)$ defines the parameters of the edge processes as a function of the latent coordinates $x, y \in [0,1]$. The edge activation probability at time $t$ for a pair of latent $(\xi_i,\xi_j)$ is given by
\begin{align*}
    \mathbb{P}[A_{ij}^{(t)} = 1 \mid& A_{ij}^{(t-1)}, A_{ij}^{(t-2)}, \xi_i, \xi_j] = \\
     &\operatorname{logistic}(b_1(\xi_i,\xi_j) A_{ij}^{(t-1)} + b_2(\xi_i,\xi_j) A_{ij}^{(t-2)} + b_{12}(\xi_i,\xi_j)A_{ij}^{(t-1)} A_{ij}^{(t-2)} +  c(\xi_i,\xi_j)),
\end{align*}
with
\[
b_1(x, y) = 4(x - 0.7)(y - 0.7), \quad b_2(x, y) = 4(x - 0.1)(y - 0.1),
\]
\[
b_{12}(x, y) = 2xy \cdot \text{sign}(x - 0.5) \cdot \text{sign}(y - 0.5),
\]
\[
c(x, y) = 10 \left(1 - \max(x, y) - \frac{3}{4}(x^2 - 1)(y^2 - 1)\right).
\]

\section{Additional Technical Considerations}
\label{appendix:technical_extensions}

\subsection{Missing data} In \Cref{subsec:hospital}, we excluded edges that exhibited no interactions, since our estimator assumes that every edge follows a well-defined Markov process (see \Cref{assumption:mc_graphon}). However, the framework can be extended to handle missing data under a missing-at-random assumption by modeling edge observations as arising from a mixture: with probability $1-p$ the edge is unobserved (missing), and with probability $p$ it is drawn from the decorated graphon. This approach permits consistent estimation provided that missingness is independent of edge values. Formally, we can model the observed edge process as
\begin{equation}
    \tilde{A}_{ij}^{(t)} = \begin{cases}
        A_{ij}^{(t)} & \text{with probability } p \\
        \text{missing} & \text{with probability } 1-p
    \end{cases}
\end{equation}
where $A_{ij}^{(t)}$ follows the decorated graphon model. Under appropriate regularity conditions on $p$ (e.g., $p$ bounded away from zero), the edge-level estimators $\est_{ij}$ can be modified to account for missingness, and the convergence rates derived in \Cref{theorem:abstract_result_conv} would remain valid up to multiplicative factors depending on $p$. The question of estimating $p$ from data and whether $p$ could depend on the latent positions (and how this would impact the error rates) is still open.

\subsection{Extension to edge-dependent concentration bounds}
\label{subsec:heterogeneous_extension}
As noted in \Cref{remark:heterogeneous_edges}, the uniform bounds in \Cref{assumption:edge_process_est} can be replaced by edge-dependent bounds $\cpsi_{ij}$ and $\bias_{ij}$, accommodating heterogeneous edge processes (e.g., Markov chains with different mixing times or state space sizes across edges). Under this relaxation, the rate in \Cref{theorem:abstract_result_blockmodel} would replace the global $\cpsi^2$ by the blockwise average
\[
\bar{\cpsi}^2 := \max_{a,b\in[k]} \frac{1}{n_a n_b}\sum_{i\in z^{-1}(a)}\sum_{j\in z^{-1}(b)} \cpsi_{ij}^2,
\]
and similarly $\bias^2$ by $\bar{\bias}^2$. The proof structure carries through: in \Cref{lemma:concentration_1}, the expectation bound $\mathbb{E}[V_{ab}(z)] \leq 2\cpsi^2$ becomes $\mathbb{E}[V_{ab}(z)] \leq 2(n_a n_b)^{-1}\sum_{i,j} \cpsi_{ij}^2$, and the sub-exponential tail bounds hold with heterogeneous norms via \citet[Prop.~5.10]{vershynin_introduction_2011} applied to sums of independent sub-Gaussians with different variance proxies. The modifications to \Cref{lemma:concentration_2} are analogous.

\subsection{Post-processing via smoothing} The block-constant structure of our estimated graphon (visible in \Cref{fig:simulation_heatmaps}) can be refined through post-processing. Recent work by \citet{verdeyme_hybrid_2024} provides smoothing techniques that can enhance estimation accuracy by mitigating discretization artifacts while preserving the overall graphon structure. Theoretical guarantees for such smoothed estimators can be derived by combining our results with those of \citet{verdeyme_hybrid_2024} as explained in \citet{dufour_inference_2024}.

\subsection{Details on connection to other temporal network models}

\subsubsection{Dynamic graphon of \citet{pensky_dynamic_2019}}  In this dynamic graphon model, the edge between nodes $i$ and $j$ is governed by a smooth function and the nodes carry the same labels at any time \citep[Sec. 7]{pensky_dynamic_2019}, and  conditionally on this function there is no dependence between two successive time observations. In our setting, we can rewrite this time varying connection probability as $W(\xi_i, \xi_j)(\cdot) = f(\xi_i, \xi_j, \cdot)$, with $W: [0,1]^2 \to L^2([0,1])$, $\xi_i \sim \mathrm{Uniform}[0,1]$ and $A_{ij}^{(t)}| \xi_i,\xi_j \sim \mathrm{Bern}(W(\xi_i, \xi_j)(t))$ independently of other entries. The continuous functions $f(\xi_i, \xi_j, \cdot)$ are assumed to have some sparse representation in an orthogonal basis \citep[Sec. 3,7]{pensky_dynamic_2019}, and thus further connect to our estimation method. However, it is worth noting that the estimation procedure in \citet{pensky_dynamic_2019} is fully adaptive to the data (through a penalization) and ours is not (as we need to pick a number of clusters $k$). Comparing the rate of convergence of the two methods could shed light on the price we pay for our two-stage estimation procedure as opposed to a joint estimation in the case of conditional independence within a single edge process.

\subsubsection{BALARM of \citet{suveges_networks_2023}}

The block-ALARM (BALARM) model of \citet{suveges_networks_2023} is a stochastic blockmodel whose edges follow a logistic autoregression. Writing $z(i)=a$ and $z(j)=b$ for the block labels,
\[
\mathbb{P}[A_{ij}^{(t)}=1 \mid A_{ij}^{(t-1)}, A_{ij}^{(t-2)}] = \operatorname{logistic}\big(\phi_{ab}(t) + b_1^{ab} A_{ij}^{(t-1)} + b_2^{ab} A_{ij}^{(t-2)} + c_{ab}\big),
\]
where $\phi_{ab}$ is a periodic term shared within the block pair; this is the process we use in the simulation of \Cref{sec:simulation_details}. It is the block-constant case of our model in which $W(\xi_i,\xi_j)$ is a periodic second-order Markov chain depending on $(\xi_i,\xi_j)$ only through $(a,b)$, so BALARM is covered by the $M$-Markov chain construction of \Cref{subsec:m_markov_chain_and_periodic} with $M=2$. The two approaches differ mainly in estimation: \citet{suveges_networks_2023} fit the autoregressive parameters jointly by likelihood for a fixed set of edge groups, whereas our two-stage procedure estimates each edge chain on its own before clustering, and applies beyond the blockmodel to the smooth-graphon regime. A more structural difference is that BALARM clusters the edge processes directly as a mixture and does not require these groups to be compatible with a partition of the nodes, while in our model the edge groups are induced by the node communities and so remain tied to the graph structure. The hospital analysis of \Cref{subsec:hospital} shows this correspondence in practice, where a decorated SBM with three node communities recovers the six BALARM edge groups.

\end{appendix}
\end{document}